\documentclass[
a4paper,
aps,
pra,
showpacs,
twocolumn,
superscriptaddress]
{revtex4-2} 

\usepackage[utf8]{inputenc}  %Input what you want e.g., é, ł, a, ü
\usepackage[T1]{fontenc}     %Output what you want e.g., é, ł, a, ü
\usepackage[american]{babel}  %Do hyphenation according to american english
\usepackage[svgnames,psnames]{xcolor}
\usepackage[colorlinks,
citecolor=FireBrick,
linkcolor=Blue,
linktocpage,
unicode]{hyperref}
\usepackage{amsmath,amssymb,amsthm,bm,mathtools,amsfonts,mathrsfs,bbm,dsfont} %Useful math packages
\newtheorem{theorem}{Theorem}
\newtheorem{observation}[theorem]{Observation}
\newcommand\norm[1]{\left\lVert#1\right\rVert}

\newcommand{\tr}{{\mathrm{tr}}}
\newcommand{\va}[1]{\ensuremath{(\Delta#1)^2}}

\newcommand{\ex}[1]{\ensuremath{\langle{#1}\rangle}}

\newcommand{\eins}{\mathbbm{1}}
\newcommand{\swap}{\mathbb{S}}
\renewcommand{\vr}{\ensuremath{\varrho}}
\renewcommand{\vec}[1]{\ensuremath{\boldsymbol{#1}}}
\usepackage{braket}
\usepackage[normalem]{ulem}

\begin{document}
\title{Collective randomized measurements in quantum information processing}

\author{Satoya Imai}
\affiliation{Naturwissenschaftlich-Technische Fakult\"at, Universit\"at Siegen, Walter-Flex-Stra{\ss}e~3, 57068 Siegen, Germany}
\affiliation{QSTAR, INO-CNR, and LENS, Largo Enrico Fermi, 2, 50125 Firenze, Italy}

\author{G\'eza T\'oth}
\affiliation{Department of Theoretical Physics, University of the Basque Country UPV/EHU, P.O. Box 644, E-48080 Bilbao, Spain}
\affiliation{EHU Quantum Center, University of the Basque Country UPV/EHU, Barrio Sarriena s/n, ES-48940 Leioa, Biscay, Spain}
\affiliation{Donostia International Physics Center (DIPC), P.O. Box 1072, E-20080 San Sebasti\'an, Spain}
\affiliation{IKERBASQUE, Basque Foundation for Science, E-48009 Bilbao, Spain}
\affiliation{HUN-REN Wigner Research Centre for Physics, P.O. Box 49, H-1525 Budapest, Hungary}

\author{Otfried G\"uhne}
\affiliation{Naturwissenschaftlich-Technische Fakult\"at, Universit\"at Siegen, Walter-Flex-Stra{\ss}e~3, 57068 Siegen, Germany}

\date{\today}
\begin{abstract}
The concept of randomized measurements on individual particles has
proven to be useful for analyzing quantum systems and is central for methods
like shadow tomography of quantum states. We introduce {\it collective} randomized
measurements as a tool in quantum information processing. Our idea is to perform
measurements of collective angular momentum on a quantum system and actively rotate
the directions using simultaneous multilateral unitaries. Based on the moments of
the resulting probability distribution, we propose systematic approaches to characterize
quantum entanglement in a collective-reference-frame-independent manner.
First, we show that existing spin-squeezing inequalities can be accessible in this scenario.
Next, we present an entanglement criterion based on three-body correlations, going beyond
spin-squeezing inequalities with two-body correlations. Finally, we apply our method to
characterize entanglement between spatially-separated two ensembles.
\end{abstract}

\maketitle

%%%%%%%%%%%%%%%%%%%%%%%%%%%%%%%%%%%%%%%%%%%%%%%%%%%%%%%%%%%%%%%%
{\it Introduction.---}Rapid advances in quantum technology have made it possible to manipulate and control increasingly complex quantum systems.
However, as the number of particles increases, the dimension of the Hilbert space grows exponentially, making it difficult to analyze quantum states fully.
One way to address this issue is to rotate measurement directions with random unitaries and consider the moments of the resulting probability distribution.
This method can provide essential quantum information about the system and give several advantages in characterizing quantum systems~\cite{elben2023randomized,cieslinski2023analysing}.

First, it allows us to obtain knowledge of the quantum state, reducing the experimental effort compared with the standard way of quantum state tomography.
Second, it is useful when some prior information about the state is not available, such as when an experiment is intended to create a particular quantum state.
Third, and most importantly, it does not need careful calibration and alignment of measurement directions or the sharing of a common frame of reference between the parties.

Several proposals have been put forward in the field of randomized measurements to
detect bipartite~\cite{brydges2019probing,elben2018renyi,ketterer2019characterizing,imai2021bound,liu2022detecting,wyderka2023probing,liu2023characterizing,wyderka2023complete} and
multipartite entanglement~\cite{tran2015quantum,tran2016correlations,ketterer2020entanglement,knips2020multipartite,ketterer2022statistically,liu2022characterizingcorr}.
Another research line of randomized measurements has
estimated several useful functions of quantum states such as
state's purity~\cite{van2012measuring},
Rényi entropies~\cite{brydges2019probing,elben2018renyi,elben2019statistical},
state's fidelities~\cite{elben2020cross},
scrambling~\cite{vermersch2019probing},
many-body topological invariants~\cite{elben2020many,cian2021many},
the von Neumann entropy~\cite{vermersch2023enhanced},
quantum Fisher information~\cite{yu2021experimental,rath2021quantum,vitale2023estimation},
and
the moments of the partially transposed quantum state~\cite{elben2020mixed,zhou2020single,yu2021optimal,neven2021symmetry,vitale2022symmetry,carrasco2022entanglement}.
Also, in the framework of shadow tomography, the techniques of randomized measurements are used to predict future measurements via estimators in data collections~\cite{aaronson2018shadow,huang2020predicting,zhang2021experimental,struchalin2021experimental,nguyen2022optimizing,rath2023entanglement}.

%%%%%%%%%%%%%%%%%%%%%%%%%%%%%%%%%%%%%%%%%%%%%%%%%%%%%%%%%%%%%%%
\begin{figure}[t]
    \centering
    \includegraphics[width=0.4\columnwidth]{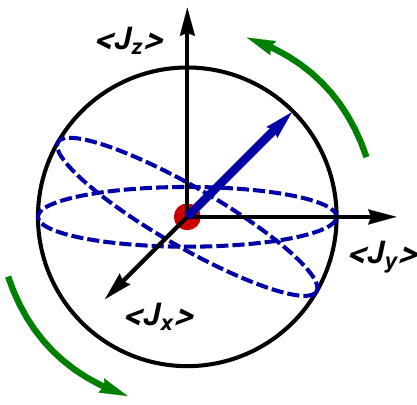}
    \caption{
    Sketch of the collective Bloch sphere with the coordinates
    $(\langle J_x\rangle, \langle  J_y\rangle, \langle J_z\rangle)$.
    Many-body spin singlet states are represented by a dot at the center (red), which does not change under any multilateral unitary transformations $U^{\otimes N}$ (green arrows).
    Spin measurement in the $z$-direction is rotated randomly (blue arrow).
    This paper proposes systematic methods to characterize spin-squeezing entanglement in an ensemble of particles
    by rotating a collective measurement direction randomly over this sphere.
    }
    \label{Fig1:sketch}
\end{figure}
%%%%%%%%%%%%%%%%%%%%%%%%%%%%%%%%%%%%%%%%%%%%%%%%%%%%%%%%%%%%%%%
Although many tools in randomized measurements were already presented,
still the current findings are not fully comprehensive in several respects.
One limitation of the results presented so far is the assumption that local
subsystems can be controlled individually.
However, this may not be available in an ensemble of quantum particles such as 
cold atoms~\cite{hald1999spin}, or
trapped ions~\cite{meyer2001experimental}, or
Bose-Einstein condensates with spin squeezing~\cite{kitagawa1993squeezed,wineland1994squeezed,sorensen2001many}.
Such quantum systems can be characterized by measuring global quantities such as collective angular momenta~\cite{toth2007optimal,toth2009spin,ma2011quantum,pezze2018quantum,braun2018quantum}.

Another practical challenge is that powerful entanglement detection 
requires many operational resources.
For instance, Refs.~\cite{imai2021bound,liu2022detecting} suggested that at
least fourth-order moments of randomized measurements are needed to characterize
a very weak form of entanglement, known as bound entanglement~\cite{horodecki1998mixed}.
In fact, their practical implementation may require 
a significant amount of randomized measurement due to the limited
availability of unitary 
designs~\cite{cieslinski2023analysing,dankert2005efficient,gross2007evenly, 
ketterer2019characterizing}.

In this paper, we generalize the concept of randomized measurements
on individual particles to the notion of {\it collective} randomized measurements.
The main idea is to perform collective random
rotations on a multiparticle quantum system before
a fixed measurement and consider the moments of the 
resulting distribution of results with respect to the
randomly chosen rotations.
We will apply this idea to different scenarios and present several entanglement
criteria in a collective-reference-frame-independent (CRFI) manner.
Note that a similar idea of collective randomization has recently been used in the context of classical shadow tomography~\cite{van2022hardware} to
classify trivial and topologically-ordered phases in many-body quantum systems.

We first show that spin-squeezing entanglement in permutationally symmetric $N$-particle systems can be characterized completely.
Second, even in non-symmetric cases, we demonstrate that the second-order moment can detect multiparticle bound entanglement.
Third, we further introduce a criterion to certify multiparticle bound entanglement with antisymmetric correlations via third-order moments.
Finally, we generalize the method to verify entanglement between spatially-separated two quantum ensembles.

%%%%%%%%%%%%%%%%%%%%%%%%%%%%%%%%%%%%%%%%%%%%%%%%%%%%%%%%%%%%%%%%
\vspace{1em}
{\it Collective randomized measurements.---}Consider a quantum ensemble that consists of $N$ spin-$\frac{1}{2}$ particles in a state $\vr \in \mathcal{H}_2^{\otimes N}$.
Suppose that each particle in this ensemble cannot be controlled individually, and one can instead measure the collective angular momentum
\begin{equation}
    J_l = \frac{1}{2}\sum_{i=1}^N \sigma_l^{(i)},
\end{equation}
with Pauli spin matrices $\sigma_l^{(i)}$ for $l=x,y,z$ acting on $i$-th subsystem.

Let us perform measurements with $J_z$ and rotate the collective direction in an arbitrary manner.
We introduce an expectation value and its variance according to a random unitary:
\begin{subequations}
\begin{align}
    \ex{J_z}_U \label{eq:Jz}
    &= \tr\left[
    \vr U^{\otimes N} J_z (U^\dagger)^{\otimes N}
    \right],\\
    \va{J_z}_U \label{eq:Jz^2}
    &= \ex{J_z^2}_U - \ex{J_z}_U^2.
\end{align}
\end{subequations}
These depend on the choice of collective simultaneous multilateral unitary operations $U^{\otimes N}$.
Now we define a linear combination as
\begin{equation}
    f_U (\vr) = \alpha \va{J_z}_U + \beta \ex{J_z}_U^2 + \gamma,
    \label{eq:funcJ}
\end{equation}
where
$\alpha, \beta, \gamma$
are real constant parameters.
The function $f_U (\vr)$ can be determined experimentally by observing 
$\ex{J_z}_U$ and $\va{J_z}_U$ as each parameter can be adjusted 
in the post-processing.

The key idea to detect entanglement in $\vr$ is to take a sample over collective local unitaries and consider the $r$-th moments of the resulting distribution
\begin{equation}
    \mathcal{J}^{(r)} (\vr)
    = \int dU \,
    [f_U (\vr)]^r, \label{eq:momentJr}
\end{equation}
where the integral is taken according to the Haar measure.
This collective unitary transformation can be written as
$U^{\otimes N} = e^{i \textit{\textbf{u}} \cdot \textit{\textbf{J}}}$,
where
$\textit{\textbf{u}} = (u_x, u_y, u_z)$ is a three-dimensional unit vector
and 
$\textit{\textbf{J}} = (J_x, J_y, J_z)$
is a vector of collective angular momenta.
The randomization of Haar collective unitaries corresponds to the uniform randomization over the three-dimensional sphere in the coordinates $(\ex{J_x}, \ex{J_y}, \ex{J_z})$.
This sphere is known as the \textit{collective Bloch sphere}~\cite{ma2011quantum, pezze2018quantum, friis2019entanglement} in an analogy of the standard Bloch sphere in a single-qubit system, illustrated in Fig.~\ref{Fig1:sketch}.

It is essential that, by definition, the moments are \textit{invariant} under any collective local unitary transformation
\begin{equation}
     \mathcal{J}^{(r)}
    \left[V^{\otimes N}  \vr (V^\dagger)^{\otimes N}
    \right]
    = \mathcal{J}^{(r)}(\vr), 
\end{equation}
for a collective local unitary $V^{\otimes N}$ for $2 \times 2$ unitaries $V$.
In the following, we will discuss CRFI entanglement detection based on the moments $\mathcal{J}^{(r)}$.

%%%%%%%%%%%%%%%%%%%%%%%%%%%%%%%%%%%%%%%%%%%%%%%%%%%%%%%%%%%%%%%%
\vspace{1em}
{\it Permutationally symmetric states.---}
To proceed, let us recall that an $N$-qubit state $\vr$ is called \textit{permutationally symmetric (bosonic)} if it satisfies $P_{ab} \vr = \vr P_{ab} = \vr$, for all $a, b \in \{1,2,\ldots, N\}$ with $a\ne b$.
Here $P_{ab}$ is an orthogonal projector onto the so-called symmetric subspace that remains invariant under all the permutations.
Note that $P_{ab}$ can be written as
$P_{ab} = (\eins + \swap_{ab})/2$
with the SWAP (flip) operator
$\swap_{ab} = \sum_{i,j} \ket{ij}\!\bra{ji}$
that can exchange qubits $a,b$:
$\swap_{ab} \ket{\psi_a}\otimes \ket{\psi_b}
= \ket{\psi_b}\otimes \ket{\psi_a}$.
We stress that the notion of permutational symmetry is stronger 
than permutational invariance, defined by $P_{ab} \vr P_{ab} = \vr$~\cite{toth2009entanglement}.

There are many studies on the entanglement of permutationally symmetric states~\cite{
eckert2002quantum,
stockton2003characterizing,
miszczak2008sub,
ichikawa2008exchange, 
toth2009entanglement,
wei2010exchange,
hansenne2022symmetries}.
In general, a state $\vr$ is said to contain multipartite entanglement 
if it cannot be written as the fully separable state
\begin{align}
    \vr_{\text{fs}}
    =
    \sum_k p_k \ket{a^{(1)}_k, a^{(2)}_k \cdots a^{(N)}_k}
    \!
    \bra{a^{(1)}_k, a^{(2)}_k \cdots a^{(N)}_k},
    \label{eq-entanglement}
\end{align}
where the $p_k$ form a probability distribution.
Importantly, for any $N$-particle permutationally 
symmetric state the pure states in a decomposition like Eq.~(\ref{eq-entanglement})
need to be symmetric, too; so a symmetric state is 
either fully separable or genuinely multipartite entangled (GME)~\cite{eckert2002quantum, wei2010exchange},
where GME states cannot be written in any separable form for all bipartitions.
One sufficient way to prove GME for a symmetric state is thus to detect entanglement
in a two-particle reduced state $\vr_{ab} = \tr_{(a,b)^c}(\vr)$ for only one pair $(a,b)$ with
the complement $(a,b)^c$. This can be achieved by accessing only the two-body correlations as minimal information.

The notion of spin squeezing originally relies on certain 
spin-squeezing parameters~\cite{wineland1992spin}, but in several previous works~\cite{wang2003spin, korbicz2005spin, korbicz2006generalized, toth2007optimal, toth2009spin}, a state $\vr$ is called \textit{spin-squeezed} if
its entanglement can be detected from the values of $\ex{J_l}$ and $\ex{J_l^2}$ only for any three orthogonal directions, e.g., $l = x,y,z$.
For a symmetric state, its spin-squeezing entanglement has been completely characterized in a \textit{necessary and sufficient} manner
by proving the entanglement in the two-particle reduced states~\cite{wang2003spin, korbicz2005spin, korbicz2006generalized, toth2007optimal, toth2009spin}.
On the other hand, such a characterization
requires optimizations over collective measurement directions for a given quantum state.

In the following, we will show that the collective randomized measurement scheme can reach the same conclusion without such an optimization.
We can formulate the first main result of this paper:
\begin{observation}\label{ob:neccesuffi}
    For an $N$-qubit permutationally symmetric state $\vr$,
    the first, second, and third moments $\mathcal{J}^{(r)}(\vr)$
    for $r=1,2,3$ completely characterize spin-squeezing entanglement.
    That is, a constructive procedure for achieving the necessary
    and sufficient condition is obtained by the moments with the parameters
    $\alpha =2/N_2,\,
    \beta =-2(N-2)/(NN_2),\,
    \gamma = -1/[2(N-1)]$
    and $N_2 = N(N-1)$.
\end{observation}

The proof of this Observation is given in Appendix~\ref{ap:neccesuffi}.
As the proof's main idea, we will first explain the known fact
that a necessary and sufficient condition is equivalent to the violation of $C \geq 0$ for the covariance matrix
$C_{ij} = \ex{\sigma_i \otimes \sigma_j}_{\vr_{ab}}
    -\ex{\sigma_i}_{\vr_a}
    \ex{\sigma_j}_{\vr_b},$
with the reduced state
$\vr_{ab}$ for any choice $a,b$, for details, see Refs.~\cite{devi2007constraints, toth2009entanglement, bohnet2016partial}.
Then we will analytically show that the violation can be determined from the moments $\mathcal{J}^{(r)}(\vr)$ for $r=1,2,3$,
and we will provide an explicit procedure to decide spin-squeezing entanglement.

We remark that any $N$-qubit permutationally symmetric state can be given by a density matrix in the so-called Dicke basis.
For $m$ excitations, the Dicke state is defined as
$\ket{{D}_{N,m}} =\binom{N}{m}^{-1/2} 
\sum_k \pi_k \big(\ket{1}^{\otimes m}\otimes \ket{0}^{\otimes (N-m)}\big)$,
where the summation is over the different permutations of the qubits
and $m$ is an integer such that $0 \leq m \leq N$.
A concrete example is the state
$|{D_{3,1}}\rangle = (\ket{001} + \ket{010} + \ket{100})/\sqrt{3}$.
Then states mixed from Dicke, W, and GHZ states are permutationally symmetric.
Accordingly, Observation~\ref{ob:neccesuffi} allows us to detect such spin-squeezed GME states in a CRFI manner.

%%%%%%%%%%%%%%%%%%%%%%%%%%%%%%%%%%%%%%%%%%%%%%%%%%%%%%%%%%%%%%%%
\vspace{1em}
{\it Multiparticle bound entanglement.---}Next, let us consider the more general case where $\vr$ is not permutationally symmetric.
Even in this case, our approach with collective randomized measurements is effective for detecting spin-squeezing entanglement.
We can present the second result in this paper:
\begin{observation}\label{ob:singletcriterion}
For an $N$-qubit state $\vr$, the first moment $\mathcal{J}^{(1)}$ with $(\alpha, \beta, \gamma) = (3,0,0)$ is given by
    \begin{align}\label{eq:sumofvariance}
    \mathcal{J}^{(1)}(\vr)
    = \sum_{l=x,y,z} \va{J_l}.
    \end{align}
Any $N$-qubit fully separable state obeys
\begin{align}
    \mathcal{J}^{(1)}(\vr)
    \geq \frac{N}{2}.
    \label{eq:singletdet}
\end{align}
Then violation implies the presence of multipartite entanglement.
\end{observation} 

The proof of this Observation is given in Appendix~\ref{ap:singletcriterion}.
The criterion in Eq.~(\ref{eq:singletdet}) itself was already established~\cite{toth2004entanglement},
so we only have to show the derivation of Eq.~(\ref{eq:sumofvariance}). 
While the proof employs Weingarten calculus~\cite{collins2022weingarten} to evaluate 
the Haar integrals, we note that the relation in Eq.~(\ref{eq:sumofvariance}) can be understood 
in the context of quantum designs. These allow to
replace unitary integrals of certain polynomials
by finite sums. In the present case, one would 
require spherical two designs~\cite{seymour1984averaging,ketterer2019characterizing,cieslinski2023analysing}.

The criterion in Eq.~(\ref{eq:singletdet}) can be maximally violated by the so-called many-body
spin singlet states $\vr_{\text{singlet}}$~\cite{eggeling2001separability, toth2004entanglement, toth2007optimal, toth2009spin, toth2010generation, urizar2013macroscopic,behbood2014generation,kong2020measurement}.
A pure singlet state 
is defined as a state invariant under any collective unitary:
$U^{\otimes N} \ket{\Psi_{\text{singlet}}}
= e^{i \theta} \ket{\Psi_{\text{singlet}}}$.
That is, it is simultaneous eigenstates of $J_l$ for $l=x,y,z$
with zero eigenvalue.
Many-body spin singlet states $\vr_{\text{singlet}}$ are mixtures of pure singlet states and are also invariant under any collective local unitary:
$U^{\otimes N} \vr_{\text{singlet}}(U^\dagger)^{\otimes N}
=\vr_{\text{singlet}}$.
Since the state $\vr_{\text{singlet}}$ has $\ex{J_l^k}= 0$ for and any integer $k$, it is at the center of the collective Bloch sphere (see Fig.~\ref{Fig1:sketch}).

Moreover, the criterion of Eq.~(\ref{eq:singletdet}) is known to be a very strong entanglement condition.
In fact, it can detect the so-called multiparticle bound entanglement~\cite{toth2007optimal,toth2009spin}, which can be positive under partial transposition (PPT) for all bipartitions~\cite{peres1996separability, horodecki2001separability}.
We stress that Observation~\ref{ob:singletcriterion} only requires second moments over Haar unitary integrals.
This shows that collective randomized measurements are fundamentally different from previous randomized measurements~\cite{imai2021bound,liu2022detecting}.

In addition, we mention that Observation~\ref{ob:singletcriterion} can be used in
probing many-body Bell nonlocality~\cite{frerot2021detecting,muller2021inferring} and
improving quantum metrology~\cite{toth2012multipartite,hyllus2012fisher,toth2014quantum}.
Also, we will discuss its high-dimensional generalizations in Appendix~\ref{ap:singletcriterion}.
Finally, we remark that the criterion in Eq.~(\ref{eq:singletdet}) is known as one of the optimal inequalities to detect spin-squeezing entanglement~\cite{toth2007optimal, toth2009spin}.

%%%%%%%%%%%%%%%%%%%%%%%%%%%%%%%%%%%%%%%%%%%%%%%%%%%%%%%%%%%%%%%%
\vspace{1em}
{\it Antisymmetric entanglement.---}So far, we have considered the moments $\mathcal{J}^{(r)} (\vr)$
based on the function $f_U(\vr)$ in Eq.~(\ref{eq:funcJ}).
Since $f_U(\vr)$ contained
two-body quantum correlations via $\ex{J_z^2}_U$ and is related to
the collective angular momenta as symmetric observables, we detected entanglement with
large two-body correlations and certain symmetries, such as Dicke and singlet states.
In the following, we will develop collective randomized measurements to analyze quantum
systems in terms of non-symmetric observables with three-body correlations.

To proceed, let us begin by considering the three-qubit observable
$\mathcal{S} (\sigma_x  \otimes  \sigma_y  \otimes  \sigma_z)
\equiv
\sigma_x  \otimes  \sigma_y  \otimes  \sigma_z
+ \sigma_y  \otimes  \sigma_z  \otimes  \sigma_x
+ \sigma_z  \otimes  \sigma_x  \otimes  \sigma_y
+ \cdots$,
  where $\mathcal{S}$ denotes the average over all permutations of indices $x,y,z$.
This observable is invariant under any particle exchange:
$\swap_{ab} \, \mathcal{S} (\sigma_x  \otimes  \sigma_y  \otimes  \sigma_z) \, \swap_{ab}
   = \mathcal{S} (\sigma_x  \otimes  \sigma_y  \otimes  \sigma_z)$,
with $\swap_{ab}$ being the SWAP operator for any $a,b$.
The $N$-qubit extension of this observable can be represented by the product of collective angular momenta
\begin{equation}
    \mathcal{O}_{\mathcal{S}}
    \equiv
    \sum_{i < j < k}
    \! \mathcal{S}
    \left(
    \sigma_x^{(i)} \otimes \sigma_y^{(j)} \otimes \sigma_z^{(k)}
    \right)
    = \frac{8}{3!}
    \mathcal{S} (J_x J_y J_z),
    \label{eq:symmetrictensor}    
\end{equation}
where
$\mathcal{S} (J_x J_y J_z) = J_x J_y J_z + J_y J_z J_x + J_z J_x J_y + \cdots$ and
$\swap_{ab} \mathcal{O}_{\mathcal{S}} \swap_{ab} = \mathcal{O}_{\mathcal{S}}$.
In general, any combination of products of collective angular momenta remains permutationally invariant under particle exchange~\cite{toth2009practical}.

An associated operator with $\mathcal{O}_{\mathcal{S}}$ from Eq.~(\ref{eq:symmetrictensor}) is the antisymmetric 
observable
\begin{align}
    \mathcal{O}_{\mathcal{A}}
    &\equiv 
    \sum_{i < j < k}
    \mathcal{A}
    \left(
    \sigma_x^{(i)} \otimes \sigma_y^{(j)} \otimes \sigma_z^{(k)}
    \right),
\end{align}
where $\mathcal{A} (\sigma_x  \otimes  \sigma_y  \otimes  \sigma_z)$ denotes the antisymmetrization
of $\sigma_x  \otimes  \sigma_y  \otimes  \sigma_z$ by taking the sum over even permutations and subtracting
the sum over odd permutations of indices $x,y,z$.
Clearly, $\mathcal{O}_{\mathcal{A}}$ cannot be constructed from collective angular momenta.

We have seen that symmetric observables based on collective angular momenta can detect symmetric entanglement with collective randomized measurements.
Then, one may wonder if antisymmetric observables such as $\mathcal{O}_{\mathcal{A}}$ can characterize antisymmetric entanglement.
Similarly to Eq.~(\ref{eq:momentJr}), we define the average over random collective local unitaries
\begin{align}
    \mathcal{T}(\vr)
    = \int dU \,
    \tr\left[
    \vr U^{\otimes N}
    \mathcal{O}_{\mathcal{A}}
    (U^\dagger)^{\otimes N}
    \right].
\end{align}
For this we can formulate the third result in this paper:
\begin{observation}\label{ob:antisymmetricdet}
The average $\mathcal{T}(\vr)$ is given by
\begin{align}
    \mathcal{T}(\vr)
    \!=\!
    \tr[\vr \mathcal{O}_{\mathcal{A}}]
    \!=\! \! \!
    \sum_{i< j< k}
    \
    \sum_{a,b,c}
    \!
    \varepsilon_{abc}
    \ex{\sigma_a^{(i)}
    \! \otimes \!
    \sigma_b^{(j)}
    \! \otimes \!
    \sigma_c^{(k)}}_\vr,
    \label{eq:mathcalTsimpleexpre}
\end{align}
where
$\varepsilon_{abc}$ denotes the Levi-Civita symbol 
for $a,b,c=x,y,z$.
Any $N$-qubit fully separable state can obey a certain tight bound
\begin{align}
    |\mathcal{T}(\vr)| \leq p_{\text{fs}}^{(N)},
    \label{eq:antisymfullsepbound}
\end{align}
where $p_{\text{fs}}^{(N)}$ can be computed analytically for $N=3$ and numerically for up to 
$N \leq 7$ and is, up to numerical precision, given by 
$p_{\text{fs}}^{(N)} = {N^2 \cot \left({\pi }/{N}\right)}/{3 \sqrt{3}}$.
Then violation implies the presence of multipartite entanglement.
\end{observation}
%%%%%%%%%%%%%%%%%%%%%%%%%%%%%%%%%%%%%%%%%%%%%%%%%%%%%%%%%%%%%%%%
The derivation of Eq.~(\ref{eq:mathcalTsimpleexpre}) and the explanation of Eq.~(\ref{eq:antisymfullsepbound})
are given in Appendix~\ref{ap:antisymmetricdet}. Also, we will analytically show that any three-qubit biseparable state
obeys $|\mathcal{T}(\vr)| \leq 2$, see Appendix~\ref{ap:antisymmetricdet}, where this violation signals GME states.

%%%%%%%%%%%%%%%%%%%%%%%%%%%%%%%%%%%%%%%%%%%%%%%%%%%%%%%%%%%%%%%
\begin{figure}[t]
    \centering
    \includegraphics[width=0.9\columnwidth]{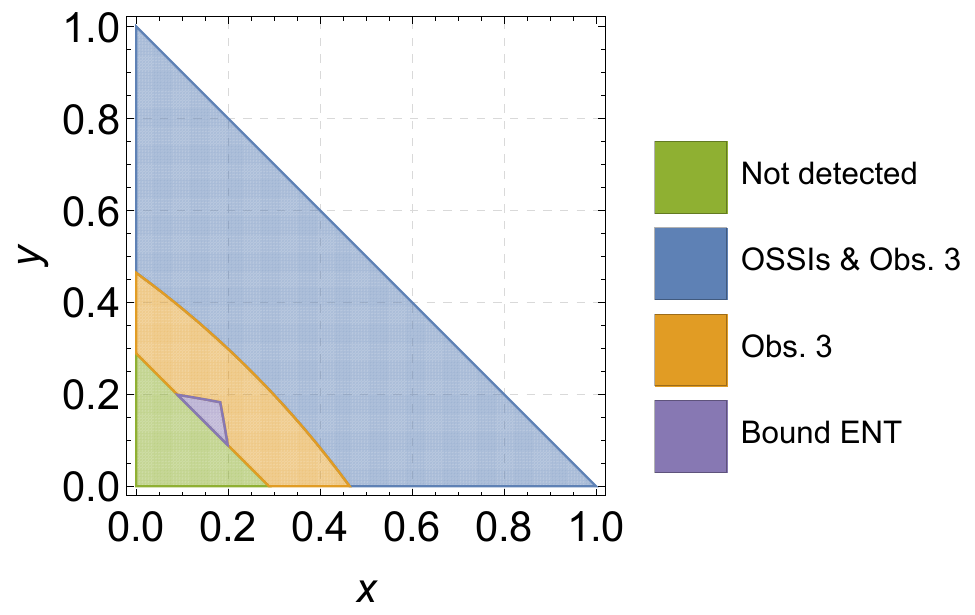}
    \caption{
    Entanglement criteria for the mixed state
    in Eq.~(\ref{eq:mixedteststate}) for $N=3$ in the $x-y$ plane.    
    The fully separable states are contained in green area,
    which obeys all the optimal spin-squeezing
    inequalities (OSSIs) previously known with optimal measurement directions~\cite{toth2007optimal, toth2009spin} and also our criterion in Obs.~\ref{ob:antisymmetricdet}.
    The blue area corresponds to the spin-squeezed entangled states that can be detected by all OSSIs and Obs.~\ref{ob:antisymmetricdet}.
    The yellow and purple areas correspond to the entangled states that cannot be detected by all OSSIs but can be detected by Obs.~\ref{ob:antisymmetricdet},
    thus marking the improvement of this paper compared with previous results.
    In particular, the purple area corresponds to the multiparticle bound
    entangled states that are not detected by the PPT criterion for all bipartitions but detected by Obs.~\ref{ob:antisymmetricdet}.}
    \label{Fig2:casethreequbit}
\end{figure}
%%%%%%%%%%%%%%%%%%%%%%%%%%%%%%%%%%%%%%%%%%%%%%%%%%%%%%%%%%%%%%%

Let us test our criterion with the two-parameter family of states
\begin{align}
    \vr_{x,y}
    = x\ket{\zeta_{N}}\! \bra{\zeta_{N}}
    + y\ket{\Tilde{\zeta}_{N}}\! \bra{\Tilde{\zeta}_{N}}
    + \frac{1-x-y}{2^N}\eins_{2}^{\otimes N},
    \label{eq:mixedteststate}
\end{align}
where
$0 \leq x, y \leq 1$.
Here $\ket{\zeta_{N}}$ is the so-called phased Dicke state \cite{krammer2009multipartite,chiuri2010hyperentangled,cramer2011measuring,marty2014quantifying,borrelli2014witnessing,marty2017multiparticle,li2021verification}, up to normalization,
\begin{align}
    \ket{\zeta_{N}}
    =
    \sum_{k=1}^N
    \frac{e^{\frac{2\pi i k}{N}}}{\sqrt{N}}
    \ket{0}_1
    \ket{0}_2
    \cdots
    \ket{1}_k
    \cdots
    \ket{0}_{N-1}
    \ket{0}_N,
\end{align}
and the state
$\ket{\Tilde{\zeta}_{N}} = \sigma_x^{\otimes N} \ket{\zeta_{N}}$ 
with $\braket{\zeta_{N}|\Tilde{\zeta}_{N}} = 0$.
Note that the phased Dicke state is not equivalent to the Dicke state $\ket{{D}_{N,1}}$ under collective unitary transformations.
In Fig.~\ref{Fig2:casethreequbit}, we illustrate the criterion of Observation~\ref{ob:antisymmetricdet} for the state $\vr_{x,y}$
for $N=3$ on the $x$-$y$ plane.
Note that the subspace spanned by $\ket{\zeta_{3}}$ and $\ket{\Tilde{\zeta}_{3}}$ is, after noncollective local
unitaries, equivalent to the maximally entangled subspace
of three qubits, as characterized in Ref.~\cite{steinberg2022maximizing}.

Our result allows us to detect entangled states that cannot be detected not only for Eq.~(\ref{eq:singletdet}) but also for all the other optimal spin-squeezing inequalities previously known with optimal measurement directions~\cite{toth2007optimal, toth2009spin}.
Moreover, the multipartite bound entanglement of $\vr_{x,y}$ can also be detected.
For $N\geq 4$, similar results are obtained, see Appendix~\ref{ap:antisymmetricdet}

The inequality (\ref{eq:antisymfullsepbound}) can be maximally
violated by several GME states.
For small $N$, we have numerically confirmed that $\ket{\zeta_{N}}$ and $\ket{\Tilde{\zeta}_{N}}$
can reach the maximal violation, that is, they can be the eigenstates with the largest singular values of $\mathcal{O}_A$.
For $N=3$, the eigenvalue decomposition of $\mathcal{O}_A$ is given by
$\mathcal{O}_A \!=\! 2\sqrt{3}
(\ket{\zeta_{3}}\! \bra{\zeta_{3}}
\!+\! \ket{\Tilde{\zeta}_{3}}\! \bra{\Tilde{\zeta}_{3}}
\!-\! \ket{\mu_{3}}\! \bra{\mu_{3}}
\!-\! \ket{\Tilde{\mu}_{3}}\! \bra{\Tilde{\mu}_{3}})$,
where
$\ket{\mu_{3}}$ is a state obtained by changing $e^{\frac{2\pi i k}{3}}$ in $\ket{\zeta_{3}}$ to $e^{\frac{2\pi i (4-k)}{3}}$
and 
$\ket{\Tilde{\mu}_{3}} = \sigma_x^{\otimes 3} \ket{\mu_{3}}$.
Here all the eigenstates are mutually orthogonal, and the dimension of this eigensubspace is four, which coincides with the maximal dimension of a three-qubit completely entangled subspace that contains no full product state~\cite{parthasarathy2004maximal,demianowicz2018unextendible}.
Finally, we mention that for cases with $N=4,5,6$, the matrix rank of $\mathcal{O}_{\mathcal{A}}$ is respectively given by $6,24,38$.

%%%%%%%%%%%%%%%%%%%%%%%%%%%%%%%%%%%%%%%%%%%%%%%%%%%%%%%%%%%%%%%%
\vspace{1em}
{\it Entanglement between two ensembles.---}Let us apply the strategy of collective randomized measurements to another scenario where two ensembles are spatially separated~
\cite{
lange2018entanglement,
fadel2018spatial,
kunkel2018spatially,
shin2019bell,
fadel2023multiparameter,
vitagliano2023number}.
We denote $\vr_{AB}$ as a $2N$-qubit state that contains the two ensembles of $N$ spin-$\frac{1}{2}$ particles,
where $\vr_{AB} \in \mathcal{H}_A \otimes \mathcal{H}_B$ with $\mathcal{H}_{X}  = \mathcal{H}_2^{\otimes N}$ for $X=A, B$.
Supposing that each ensemble can be controlled individually, we can perform the collective randomized measurements to obtain the moments $\mathcal{J}_X^{(r)}$ with a fixed choice $(\alpha, \beta, \gamma)$.
Note that a related approach to detect entanglement between two spin ensembles has been discussed in Ref.~\cite{kofler2006entanglement}.

The total collective observables are given by
$J_{l}^{\pm} = J_{l, A} \pm J_{l, B}$,
where
$J_{l, X} = \frac{1}{2}\sum_{i=1}^N \sigma_{l}^{(X_i)} \in \mathcal{H}_{X}$
for $l=x,y,z$ and Pauli matrices $\sigma_{l}^{(X_i)}$ acting on $X_i$-th subsystem in the ensemble $X=A,B$.
Note that one can also formulate the following result considering a more general case: $J_{k,l}^{\pm} = J_{k, A} \pm J_{l, B}$.
In a similar manner to Eq.~(\ref{eq:Jz^2}), we can introduce the random variances
$\va{J_{z}^{\pm}}_{U_{AB}}$ with
$U_{AB} = U_A \otimes U_B$.
Denoting the gap as
$\eta_{U_{AB}}
\equiv \va{J_{z}^{+}}_{U_{AB}} - \va{J_{z}^{-}}_{U_{AB}}$,
let us consider its moment
\begin{align}
    \mathcal{G}_{AB}^{(r)}
    = g\int dU_{AB}\,
    [\eta_{U_{AB}}]^r,
\end{align}
where $g$ is a real constant parameter.

Now we can present the following criterion:
%%%%%%%%%%%%%%%%%%%%%%%%%%%%%%%%%%%%%%%%%%%%%%%%%%%%
\begin{observation}\label{ob:twoensembles}
For a $2N$-qubit state $\vr_{AB}$
with the permutationally symmetric reduced states,
any separable $\vr_{AB}$ obeys
\begin{align}
    \mathcal{G}_{AB}^{(2)}
    + \mathcal{J}_A^{(1)}
    + \mathcal{J}_B^{(1)}
    - \mathcal{J}_A^{(1)}
    \mathcal{J}_B^{(1)}
    \leq 1,
    \label{eq:twoensembledetection}
\end{align}
where
$g = ({3}/{N^2})^2$
and 
$(\alpha, \beta, \gamma) = (0,12/N^2,0)$.
\end{observation}
%%%%%%%%%%%%%%%%%%%%%%%%%%%%%%%%%%%%%%%%%%%%%%%%%%%%
The proof is given in Appendix~\ref{ap:twoensembles}.
As the proof's main idea, we will first simply evaluate the integrals on the left-hand side in Eq.~(\ref{eq:twoensembledetection}).
Then we will adopt the separability criterion presented in Ref.~\cite{liu2022characterizingcorr} (see, Proposition 5 there) in order to find the entanglement criterion in Eq.~(\ref{eq:twoensembledetection}).

The violation of this inequality allows us to detect entanglement between the spatially separated two ensembles.
In Appendix~\ref{ap:twoensembles},
we will demonstrate how the criterion in Eq.~(\ref{eq:twoensembledetection})
can characterize entanglement between two ensembles.
Also, we will show that Observation~\ref{ob:twoensembles} can be extended to the case of $m$ ensembles for $m\geq 3$.

%%%%%%%%%%%%%%%%%%%%%%%%%%%%%%%%%%%%%%%%%%%%%%%%%%%%%%%%%%%%%%%
\vspace{1em}
{\it Statistically significant tests.---}Finally, we note that the statistical analysis of collective randomized measurements is discussed in Appendix~\ref{ap:finitestatistics}. There, we will provide estimations for the necessary number of measurements required for entanglement detection with high confidence. Similar discussions can also be found in Refs.~\cite{ketterer2022statistically,wyderka2023probing,liu2023characterizing,wyderka2023complete,cieslinski2023analysing}.

%%%%%%%%%%%%%%%%%%%%%%%%%%%%%%%%%%%%%%%%%%%%%%%%%%%%%%%%%%%%%%%
\vspace{1em}
{\it Conclusion.---}We have introduced the concept of collective randomized measurements as a tool for CRFI quantum information processing.
Based on the framework, we have proposed systematic methods to characterize quantum correlations.
In particular, we showed that our approach has detected spin-squeezing entanglement, multipartite bound entanglement, and spatially-separated entanglement in two ensembles.

There are several directions for future research. First, it would be interesting to extend 
our work to higher-order scenarios involving $J_l^k$ for $k>2$. Such extensions may 
facilitate various connections, such as nonlinear spin 
squeezing~\cite{gessner2019metrological} or permutationally invariant Bell 
inequalities~\cite{tura2014detecting,tura2015nonlocality,wagner2017bell, guo2023detecting}.
Next, the inequality (\ref{eq:antisymfullsepbound}) resembles multipartite 
entanglement witnesses~\cite{guhne2009entanglement}, exploring this may 
lead to more advanced techniques for analyzing multipartite entanglement. Finally, while the 
standard version of randomized measurements has found many applications beyond entanglement detection, e.g., in quantum metrology, shadow tomography, or cross-platform verification~\cite{elben2023randomized,cieslinski2023analysing}, one may study
these possibilities also for the collective randomized measurements introduced here.

%%%%%%%%%%%%%%%%%%%%%%%%%%%%%%%%%%%%%%%%%%%%%%%%%%%%
\vspace{1em}
{\it Acknowledgments.---}We would like to thank
Iagoba Apellaniz,
Jan Lennart B\" onsel,
Sophia Denker,
Kiara Hansenne,
Andreas Ketterer,
Matthias Kleinmann,
Yi Li,
Zhenhuan Liu,
H. Chau Nguyen,
Stefan Nimmrichter,
Martin Pl\'avala,
R\'obert Tr\'enyi,
Giuseppe Vitagliano,
Nikolai Wyderka,
Zhen-Peng Xu,
Benjamin Yadin,
and
Xiao-Dong Yu
for discussions.

This work was supported by
the DAAD,
the Deutsche Forschungsgemeinschaft (DFG, German Research Foundation, project numbers 447948357 and 440958198),
the Sino-German Center for Research Promotion (Project M-0294),
the ERC (Consolidator Grant 683107/TempoQ),
the German Ministry of Education and Research (Project QuKuK, BMBF Grant No. 16KIS1618K),
the Hungarian Scientific Research Fund (Grant No. 2019-2.1.7-ERA-NET-2021-00036),
and the National Research, Development and Innovation Office of Hungary (NKFIH) within the Quantum Information National Laboratory of Hungary.
We acknowledge the support of the EU 
%(COST Action CA15220, QuantERA CEBBEC,
(QuantERA MENTA, QuantERA QuSiED),
the Spanish MCIU 
%(Grant No. PCI2018-092896, 
(No. PCI2022-132947), and the Basque Government 
%(Grant No. IT986-16, 
(No. IT1470-22). 
We acknowledge the support of the Grant~No.~PID2021-126273NB-I00 funded by MCIN/AEI/10.13039/501100011033 and by "ERDF A way of making Europe".  
We thank the "Frontline" Research Excellence Programme of the NKFIH (Grant No. KKP133827).
G.~T. acknowledges a  Bessel Research Award of the Humboldt Foundation.
We acknowledge Project no. TKP2021-NVA-04, which has been implemented with the support provided by the Ministry of Innovation and Technology of Hungary from the National Research, Development and Innovation Office (NKFIH), financed under the TKP2021-NVA funding scheme. 

%%%%%%%%%%%%%%%%%%%%%%%%%%%%%%%%%%%%%%%%%%%%%%%%%%%%
\onecolumngrid
\appendix
\addtocounter{theorem}{-4}
%%%%%%%%%%%%%%%%%%%%%%%%%%%%%%%%%%%%%%%%%%%%%%%%%%%%%
\section{Proof of Observation~\ref{ob:neccesuffi}}\label{ap:neccesuffi}
\begin{observation}
    For an $N$-qubit permutationally symmetric state $\vr$,
    the first, second, and third moments $\mathcal{J}^{(r)}(\vr)$
    for $r=1,2,3$ completely characterize spin-squeezing entanglement.
    That is, a constructive procedure for achieving the necessary
    and sufficient condition is obtained by the moments with the parameters
    $\alpha =2/N_2,\,
    \beta =-2(N-2)/(NN_2),\,
    \gamma = -1/[2(N-1)]$
    and $N_2 = N(N-1)$.
\end{observation}

%%%%%%%%%%%%%%%%%%%%%%%%%%%%%%%%%%%%%%%%%%%%%%%%%%%%%
\begin{proof}
In the following, we will first describe the logic of
how to prove Observation~\ref{ob:neccesuffi} in the
main text and later
explain each line step-by-step
\begin{subequations}
\begin{align}
    \vr_{\text{PS}} 
    \in \mathcal{H}_2^{\otimes N}
    \ \text{is}
    \ \text{spin squeezed}
    &\iff
    \vr_{ab} \in \mathcal{H}_2^{\otimes 2}
    \ \text{is}
    \ \text{entangled}
    \\
    &\iff
    \vr_{ab} \not\in
    \text{PPT}\\
    &\iff
    M \ngeq 0\\
    &\iff
    C \ngeq 0\\
    &\iff
    \text{obtained}
    \ \text{from}
    \
    \tr(C^r)
    \     \text{for}
    \
    r=1,2,3\\
    &\iff
    \text{obtained}
    \ \text{from}
    \
    \mathcal{C}^{(r)}(\vr)
    \     \text{for}
    \
    r=1,2,3\\
    &\iff
    \text{obtained}
    \ \text{from}
    \
    \mathcal{J}^{(r)}(\vr)
    \     \text{for}
    \
    r=1,2,3.
\end{align}
\end{subequations}
In the first line, we denote an $N$-qubit permutationally symmetric
state as $\vr_{\text{PS}}$ and recall again that it possesses bipartite entanglement or often spin squeezing if and only if any two-qubit reduced state $\vr_{ab} = \tr_{(a,b)^c} (\vr_{\text{PS}})$ is
entangled for $a,b=1,2,\ldots, N$, where $X^c$ is the complement of
a set $X$.
This has been already discussed in Refs.~\cite{wang2003spin, korbicz2005spin, korbicz2006generalized}.
In the second line, we also recall that any two-qubit state is entangled if and only if it has a negative eigenvalue under partial transposition, that is, it violates
the so-called PPT criterion~\cite{peres1996separability, horodecki2001separability}.

In the third line, we first recall that any two-qubit state $\vr_{ab}$ can
be written as
\begin{align}
    \vr_{ab} = \frac{1}{4}\sum_{i,j=0}^{3} m_{ij}\sigma_i \otimes \sigma_j.
\end{align}
Here we note that a two-qubit state $\vr_{ab}$ is permutationally symmetric
and separable (that is, PPT) if and only if it holds that $M \geq 0$, where $M = (m_{ij})$ for $i,j=0,1,2,3$.
In the fourth line, this separability condition is equivalent to $C \geq 0$
for a permutationally symmetric $\vr_{ab}$.
Here the $3 \times 3$ matrix $C = (C_{ij})$ is the Schur complement of the $4 \times 4$ matrix $M$,
which is given by $C_{ij} = m_{ij} - m_{i0} m_{0j}$ for $i,j=1,2,3$ since $m_{00}=1$.
For details, see Refs.~\cite{devi2007constraints, toth2009entanglement, bohnet2016partial}.

In the fifth line, we first discuss the explicit form of the covariance matrix $C$
\begin{align} \label{eq:covappA}
    C_{ij}=
    \tr[\vr_{ab} \sigma_i \otimes \sigma_j]
    -\tr[\vr_{a} \sigma_i]
    \tr[\vr_{b} \sigma_j]
    = t_{ij} - a_i a_j,
\end{align}
where
$m_{ij} = t_{ij} = t_{ji}$ and
$m_{i0} = m_{0i} = a_i$
since $\vr_{ab}$ is permutationally symmetric.
Then, the covariance matrix $C = T  - \vec{a} \vec{a}^\top$ is symmetric $C = C^\top$,
where
$T=(t_{ij}) = T^\top$ with the constraint $\tr[T]=\sum_i t_{ii} = 1$
and $\Vec{a} = (a_x, a_y, a_z)$.
To proceed, let us remark that the matrix $C$ can be diagonalized by a collective
local unitary transformation $V \otimes V$, leads to that
$OCO^\top = \text{diag}(c_1, c_2, c_3)$ with a rotation matrix $O \in SO(3)$.
In fact, the eigenvalues $c_1, c_2, c_3$ can be found by computing the roots of
the characteristic polynomial
\begin{align}
    p_C(\lambda) = \lambda^3 - \tr(C)\lambda^2 + \frac{1}{2}\left[
    \tr(C)^2 -\tr(C^2)    \right] \lambda
    -\det{(C)},
\end{align}
where $\tr(C^r) = \sum_{i=1,2,3} c_i^r$ and the $\det{(C)}$ can be written as
\begin{align}
    \det{(C)} = \frac{1}{6}\left[
    \tr(C)^3 -3\tr(C)\tr(C^2) + 2\tr(C^3)
    \right].
\end{align}
That is, knowing the $\tr[C^r]$ for $r=1,2,3$ can enable us to
access its eigenvalues and therefore decide whether the matrix $C$ is positive
or negative.

In the sixth and seventh lines, it is sufficient to show that $\tr[C^r]$ for $r=1,2,3$ can be
obtained from the moments $\mathcal{J}^{(r)}(\vr)$ in the collective randomized
measurements.
For the choice 
$\alpha = 2/N_2$,
$\beta = -2(N-2)/(NN_2)$,
$\gamma = -1/[2(N-1)]$,
and $N_2 = N(N-1)$,
we immediately find that the moments $\mathcal{J}^{(r)}(\vr)$ can be equal to the moments $\mathcal{C}^{(r)} (\vr_{ab})$ of the random covariance matrix
\begin{subequations}
\begin{align}
        \mathcal{J}^{(r)}(\vr)
        &=
        \mathcal{C}^{(r)} (\vr_{ab})
        \equiv \int dU \, [\text{Cov}_U]^r,\\
        \text{Cov}_U
        &=\tr[\vr_{ab} U^{\otimes 2}\sigma_z \otimes \sigma_z
        (U^\dagger)^{\otimes 2}]
        -\tr[\vr_{a} U\sigma_z U^\dagger]
        \tr[\vr_{b} U\sigma_z U^\dagger].
\end{align}
\end{subequations}
This results from the fact that
$\ex{J_z}_U = \frac{N}{2} \tr[\vr_{a} U\sigma_z U^\dagger]$
and
$\ex{J_z^2}_U
= \frac{N}{4}
+ \frac{N(N-1)}{2}
\tr[\vr_{ab} U^{\otimes 2}\sigma_z \otimes \sigma_z
(U^\dagger)^{\otimes 2}]$.
In the following, we will evaluate the moments $\mathcal{C}^{(r)} (\vr_{ab})$ and show that they are associated with $\tr[C^r]$.

Let us begin by rewriting the moments $\mathcal{C}^{(r)}(\vr_{ab})$ as
\begin{align} \label{momentexpr}
    \mathcal{C}^{(r)}(\vr_{ab})
    =  \frac{1}{4^r}
    \int dU \, \left[
    \sum_{i,j=x,y,z} C_{ij} \,
    \mathcal{O}_U^{(i)}\mathcal{O}_U^{(j)}
    \right]^r,
\end{align}
where we define that
$\mathcal{O}_U^{(i)} = \tr[\sigma_i U\sigma_z U^\dagger]$.
For convenience, we denote that 
\begin{subequations}
\begin{align}
    &\mathcal{I}^{(1)}(i,j)
    = \int dU \, \mathcal{O}_U^{(i)}\mathcal{O}_U^{(j)},\\
    &\mathcal{I}^{(2)}(i,j,k,l)
    = \int dU \,
    \mathcal{O}_U^{(i)}\mathcal{O}_U^{(j)}
    \mathcal{O}_U^{(k)}\mathcal{O}_U^{(l)},\\
    &\mathcal{I}^{(3)}(i,j,k,l,m,n)
    = \int dU \,
    \mathcal{O}_U^{(i)}\mathcal{O}_U^{(j)}
    \mathcal{O}_U^{(k)}\mathcal{O}_U^{(l)}
    \mathcal{O}_U^{(m)}\mathcal{O}_U^{(n)}.
\end{align}
\end{subequations}
These integrals are known to be simplified as follows
\begin{subequations}
    \begin{align}
    &\mathcal{I}^{(1)}(i,j) \label{eq:haarijappa}
    = \frac{4}{3} \delta_{ij},\\
    &\mathcal{I}^{(2)}(i,j,k,l)
    =\frac{16}{15}
    \left(
    \delta_{ij}\delta_{kl} 
    + \delta_{ik}\delta_{jl}
    + \delta_{il}\delta_{jk}
    \right),\\
    &\mathcal{I}^{(3)}(i,j,k,l,m,n) \nonumber
    =\frac{64}{105}\Big\{
     \delta_{ij} [\delta_{kl}\delta_{mn} +\delta_{km}\delta_{ln}+\delta_{kn}\delta_{lm}]
    +\delta_{ik} [\delta_{jl}\delta_{mn} +\delta_{jm}\delta_{ln}+\delta_{jn}\delta_{lm}]\\ \nonumber
    & \qquad \qquad \qquad \qquad \qquad \ 
    +\delta_{il} [\delta_{jk}\delta_{mn} +\delta_{jm}\delta_{kn}+\delta_{jn}\delta_{km}]
    +\delta_{im} [\delta_{jk}\delta_{ln} +\delta_{jl}\delta_{kn}+\delta_{jn}\delta_{kl}]\\
    &\qquad \qquad \qquad \qquad \qquad \ 
    +\delta_{in} [\delta_{jk}\delta_{lm} +\delta_{jl}\delta_{km}+\delta_{jm}\delta_{kl}]
    \Big\},
\end{align}
\end{subequations}
where we use the formulas in Ref.~\cite{wyderka2023complete}.
Substituting the above results into the moments in Eq.~(\ref{momentexpr}) for different $r=1,2,3$, we can straightforwardly obtain the following expressions
\begin{subequations}
    \begin{align}
    \mathcal{C}^{(1)}(\vr_{ab})
    &= \frac{1}{3}\tr(C),\\
    \mathcal{C}^{(2)}(\vr_{ab})
    &= \frac{1}{15}\left[
    \tr(C)^2 + \tr(CC^\top) + \tr(C^2)
    \right],\\
    \mathcal{C}^{(3)}(\vr_{ab})
    &= \frac{1}{105}\left\{
    \tr(C)\left[\tr(C)^2+ 3\tr(C^2) + 3\tr(CC^\top)
    \right]
    + 4\tr(C^2C^\top) + 2\tr(CC^\top C^\top)
    + 2\tr(C^3)
    \right\}.
\end{align}
\end{subequations}
Furthermore, using the symmetric condition $C = C^\top$, we can finally arrive at
\begin{subequations}
    \begin{align}
    \mathcal{C}^{(1)}(\vr_{ab})
    &= \frac{1}{3}\tr(C),\\
    \mathcal{C}^{(2)}(\vr_{ab})
    &= \frac{1}{15}\left[
    \tr(C)^2 + 2 \tr(C^2)
    \right],\\
    \mathcal{C}^{(3)}(\vr_{ab})
    &= \frac{1}{105}\left\{
    \tr(C)\left[\tr(C)^2+ 6\tr(C^2) \right]
    + 8\tr(C^3)
    \right\}.
\end{align}
\end{subequations}
The moments $\mathcal{C}^{(r)}(\vr_{ab})$,
equivalently $\mathcal{J}^{(r)}(\vr)$, are directly connected to $\tr[C^r]$.
Hence we complete the proof.   
\end{proof}

%%%%%%%%%%%%%%%%%%%%%%%%%%%%%%%%%%%%%%%%%%%%%%%%%%%%%
\section{Derivation of the result in Observation~\ref{ob:singletcriterion}}\label{ap:singletcriterion}
\begin{observation}
For an $N$-qubit state $\vr$, the first moment $\mathcal{J}^{(1)}$ with $(\alpha, \beta, \gamma) = (3,0,0)$ is given by
    \begin{align}\label{eq:sumofvarianceB}
    \mathcal{J}^{(1)}(\vr)
    = \sum_{l=x,y,z} \va{J_l}.
    \end{align}
Any $N$-qubit fully separable state obeys
\begin{align}
    \mathcal{J}^{(1)}(\vr)
    \geq \frac{N}{2}.
\end{align}
Then violation implies the presence of multipartite entanglement.
\end{observation} 

%%%%%%%%%%%%%%%%%%%%%%%%%%%%%%%%%%%%%%%%%%%%%%%%%%%%%
\begin{proof}
Here we give the derivation of Eq.~(\ref{eq:sumofvarianceB}).
Let us begin by writing that
$\mathcal{J}^{(1)} (\vr) =3 \int dU \, \va{J_z}_{U}$
and 
\begin{align} \nonumber
    \va{J_z}_{U}
    &= \ex{
    U^{\otimes N}J_z^2 (U^{\dagger})^{\otimes N}
    }_\vr
    -\ex{
    U^{\otimes N}J_z (U^{\dagger})^{\otimes N}
    }_{\vr}^2\\ 
    \nonumber
    &=\frac{1}{4}\sum_{i,j=1}^N
    \ex{
    U^{\otimes N} \sigma_z^{(i)} \otimes \sigma_z^{(j)} (U^{\dagger})^{\otimes N}
    }_\vr
    -\frac{1}{4}\sum_{i,j=1}^N
    \ex{
    U^{\otimes N}\sigma_z^{(i)} (U^{\dagger})^{\otimes N}
    }_{\vr}
    \ex{
    U^{\otimes N}\sigma_z^{(j)} (U^{\dagger})^{\otimes N}
    }_{\vr}\\
    &=\frac{1}{4}
    \left\{N+
    \sum_{i \neq j}^N 
    \tr \left[
    U^{\otimes 2} \sigma_z^{(i)} \otimes \sigma_z^{(j)} (U^{\dagger})^{\otimes 2}
    \vr_{ij}
    \right]
    -\sum_{i,j=1}^N
    \tr \left[
    U\sigma_z^{(i)} U^{\dagger}
    \vr_i
    \right]
    \tr \left[
    U\sigma_z^{(j)} U^{\dagger}
    \vr_j
    \right]
    \right\},
    \label{eq:vaJzUeval}
\end{align}
where $\vr_{ij}$ and $\vr_{i}$ are the two-qubit and single-qubit reduced states of $\vr$.
Let us focus on the second term in Eq.~(\ref{eq:vaJzUeval}) and take the Haar unitary average
\begin{align}\nonumber
    \sum_{i \neq j}^N \int dU \,
    \tr \left[
    U^{\otimes 2} \sigma_z^{(i)}\otimes \sigma_z^{(j)} (U^{\dagger})^{\otimes 2}
    \vr_{ij}
    \right]
    &=\frac{1}{4} \sum_{i \neq j}^N \int dU \,
    \tr \left[
    U^{\otimes 2} \sigma_z^{(i)}\otimes \sigma_z^{(j)} (U^{\dagger})^{\otimes 2}
    \sum_{k,l=x,y,z} t_{kl}^{(i,j)}
    \sigma_k^{(i)}\otimes \sigma_l^{(j)}
    \right]\\
    \nonumber
    &=\frac{1}{4} \sum_{i \neq j}^N
    \sum_{k,l=x,y,z} t_{kl}^{(i,j)}
    \int dU \,
    \tr \left[
    U\sigma_z^{(i)}U^{\dagger}
    \sigma_k^{(i)}
    \right]
    \tr \left[
    U\sigma_z^{(j)}U^{\dagger}
    \sigma_l^{(j)}
    \right]\\
    &=\frac{1}{3} \sum_{i \neq j}^N
    \sum_{l=x,y,z} t_{ll}^{(i,j)}
    =\frac{4}{3}  \sum_{l=x,y,z} \ex{J_l^2} - N,
\end{align}
where
$t_{kl}^{(i, j)} = \ex{\sigma_k^{(i)}\otimes \sigma_l^{(j)}}_{\vr_{ij}}$.
Similarly, the third term in Eq.~(\ref{eq:vaJzUeval}) can be given by
\begin{align}
    \sum_{i,j=1}^N
    \int dU \,
    \tr \left[
    U\sigma_z^{(i)} U^{\dagger}
    \vr_i
    \right]
    \tr \left[
    U\sigma_z^{(j)} U^{\dagger}
    \vr_j
    \right]
    =\frac{4}{3} \sum_{l=x,y,z} \ex{J_l}^2.
\end{align}
Summarizing these results, we can thus arrive at
\begin{align}
    \mathcal{J}^{(1)} (\vr)
    =\frac{3}{4}
    \left\{
    N+
    \frac{4}{3}  \sum_{l=x,y,z} \ex{J_l^2} - N
    + \frac{4}{3} \sum_{l=x,y,z} \ex{J_l}^2
    \right\} =  \sum_{l=x,y,z} \va{J_l}.
\end{align}
\end{proof}

\noindent
\textbf{Remark B1.}
Here we consider the generalization of Observation~\ref{ob:singletcriterion} in the main text to
high-dimensional spin systems.
For that, let us denote the $N$-qudit collective operators
$\Lambda_l = \frac{1}{d}\sum_{i=1}^N \lambda_l^{(i)}$,
with the so-called Gell-Mann matrices $\lambda_l^{(i)}$ for $l=1,2,\ldots,d^2-1$ acting on $i$-th system.
The Gell-Mann matrices are $d$-dimensional extensions of Pauli matrices satisfying the properties:
$\lambda_l^\dagger = \lambda_l,
\tr(\lambda_l) = 0,
\tr(\lambda_k \lambda_l) = d\delta_{kl}$.
For details, see~\cite{kimura2003bloch, bertlmann2008bloch,eltschka2018distribution,wei2022antilinear}.
Let us define the random expectation and its variance
\begin{align}
    \ex{\Lambda_l}_U
    = \tr[\vr U^{\otimes N}
    \Lambda_l (U^\dagger)^{\otimes N}],
    \qquad
    \va{\Lambda_l}_U
    =\ex{\Lambda_l^2}_U- \ex{\Lambda_l}_U^2,
\end{align}
which depends on the choice of collective unitaries $U^{\otimes N}$ with $U \in \mathcal{U}(d)$.
Now we introduce the average of $\va{\Lambda_l}_U$ for any $l$ over Haar random unitaries
\begin{align}
    \mathcal{D}(\vr)
    =(d^2-1)
    \int dU \, \va{\Lambda_l}_U.
\end{align}
Now we can make the following:

\vspace{1em}
\noindent
\textbf{Remark B2.}
\textit{For an $N$-qudit state $\vr$, the average can be given by
\begin{align}
    \mathcal{D}(\vr)
    = \sum_{l=1}^{d^2-1} \va{\Lambda_l}.
    \label{eq:appendixbevaluation}
\end{align}
Any $N$-qudit fully separable state obeys
\begin{align}
    \sum_{l=1}^{d^2-1} \va{\Lambda_l}
    \geq \frac{N(d-1)}{d}.
\end{align}
This violation implies the presence of multipartite entanglement.
}

This is the generalization of Observation~\ref{ob:singletcriterion}.
The fully separable bound was already discussed in Refs.~\cite{vitagliano2011spin,vitagliano2014spin}.
In the following, we give the derivation of Eq.(\ref{eq:appendixbevaluation}).

\begin{proof}
Similarly to Eq.~(\ref{eq:vaJzUeval}), the random variance $\va{\Lambda_l}_{U}$ can be written as
\begin{align}
    \va{\Lambda_l}_{U}
    \!=\! \frac{1}{d^2}\!
    \left\{\sum_{i=1}^N 
    \tr[U (\lambda_l^{(i)})^2  U^\dagger \vr_i]
    \!+\!
    \sum_{i \neq j}^N 
    \tr \left[
    U^{\otimes 2} \lambda_l^{(i)} \otimes \lambda_l^{(j)} (U^{\dagger})^{\otimes 2}
    \vr_{ij}
    \right]
    \!-\!
    \sum_{i,j=1}^N
    \tr \left[
    U\lambda_l^{(i)} U^{\dagger}
    \vr_i
    \right]
    \tr \left[
    U\lambda_l^{(j)} U^{\dagger}
    \vr_j
    \right]
    \right\},
    \label{eq:Lambdavariance}
\end{align}
where $\vr_{ij}$ and $\vr_{i}$ are the two-qudit and single-qudit reduced states of $\vr$.
To evaluate the Haar unitary integral, let us use the known formulas~\cite{cieslinski2023analysing,dankert2005efficient,zhang2014matrix,roberts2017chaos,garcia2021quantum,brandao2021models,imai2023work}
\begin{align}
\int dU\, U X U^\dagger
=\frac{\tr[X]}{d} \eins_d,
\quad
\int dU\, U^{\otimes 2} X (U^\dagger)^{\otimes 2}
=\frac{1}{d^2-1}
    \left\{
    \left[\tr(X)-\frac{\tr(\swap X)}{d}\right] \eins_d^{\otimes 2}
    +\left[\tr(\swap X)-\frac{\tr(X)}{d}
    \right]\swap
    \right\},
    \label{eq:haarintegralformula}
\end{align}
for an operator $X$.
Here $\swap$ is the SWAP (flip) operator acting on $d \times d$-dimensional systems, defined as
$\swap \ket{a}\ket{b}= \ket{b}\ket{a}$.
Thus we first obtain
\begin{align}
    \sum_{i=1}^N 
    \int dU\, \tr[U (\lambda_l^{(i)})^2  U^\dagger \vr_i]
    = \sum_{i=1}^N 
    \frac{\tr[(\lambda_l^{(i)})^2]}{d}
    \tr[\vr_i]
    = N.
\end{align}
Next, we have
\begin{align} \nonumber
    \sum_{i \neq j}^N \int dU \,
    \tr \left[
    U^{\otimes 2} \lambda_l^{(i)}\otimes \lambda_l^{(j)} (U^{\dagger})^{\otimes 2}
    \vr_{ij}
    \right]
    &=\frac{1}{d^2}
    \sum_{i \neq j}^N \int dU \,
    \tr \left[
    U^{\otimes 2} \lambda_l^{(i)}\otimes \lambda_l^{(j)} (U^{\dagger})^{\otimes 2}
    \sum_{m,n=1}^{d^2-1}
    t_{mn}^{(i,j)}
    \lambda_m^{(i)}\otimes
    \lambda_n^{(j)}
    \right]\\
    \nonumber
    &=\frac{1}{d^2}\frac{1}{d^2-1}
    \sum_{i \neq j}^N
    \sum_{m,n=1}^{d^2-1} t_{mn}^{(i,j)}
    \tr \left[
    (d\swap - \eins_d^{\otimes 2})
    \lambda_m^{(i)}\otimes
    \lambda_n^{(j)}
    \right]\\
    &=\frac{1}{d^2-1} \sum_{i \neq j}^N
    \sum_{l=1}^{d^2-1}  t_{ll}^{(i,j)}
    \nonumber
    \\
    &=\frac{d^2}{d^2-1}
    \sum_{l=1}^{d^2-1} \ex{\Lambda_l^2} - N.
\end{align}
In the first line, we denote that
$t_{mn}^{(i, j)} = \ex{\lambda_m^{(i)}\otimes \lambda_n^{(j)}}_{\vr_{ij}}$.
In the second line, we used the formula in Eq.~(\ref{eq:haarintegralformula}) and the so-called SWAP trick:
$\tr[\swap X] = \tr[\swap (X_A\otimes X_B)] = \tr[X_A X_B]$ for an operator $X=X_A \otimes X_B $.
In the final line, we used that
$\sum_{l=1}^{d^2-1} \lambda_l^2 = (d^2-1)\eins_d$, which can be derived from the facts that
$\swap = \frac{1}{d} \sum_{l=0}^{d^2-1}
\lambda_l \otimes \lambda_l$
and 
$\swap^2 = \eins_d^{\otimes 2}$.
Similarly, we obtain
\begin{align}
    \sum_{i,j=1}^N
    \int dU \,
    \tr \left[
    U\lambda_l^{(i)} U^{\dagger}
    \vr_i
    \right]
    \tr \left[
    U\lambda_l^{(j)} U^{\dagger}
    \vr_j
    \right]
    =\frac{1}{d^2-1}
    \sum_{i,j=1}^N
    \left[d\tr(\vr_i \vr_j)-1\right]
    =\frac{d^2}{d^2-1}
    \sum_{l=1}^{d^2-1} \ex{\Lambda_l}^2.
\end{align}
Summarizing these results, we can complete the proof.
\end{proof}

%%%%%%%%%%%%%%%%%%%%%%%%%%%%%%%%%%%%%%%%%%%%%%%%%%%%%
\section{Detailed discussion of Observation~\ref{ob:antisymmetricdet}
}\label{ap:antisymmetricdet}
%%%%%%%%%%%%%%%%%%%%%%%%%%%%%%%%%%%%%%%%%%%%%%%%%%%%%
\begin{observation}
The average $\mathcal{T}(\vr)$ is given by
\begin{align}
    \mathcal{T}(\vr)
    =
    \tr[\vr \mathcal{O}_{\mathcal{A}}]
    =
    \sum_{i< j< k}
    \
    \sum_{a,b,c}
    \varepsilon_{abc}
    \ex{\sigma_a^{(i)}
     \otimes 
    \sigma_b^{(j)}
     \otimes 
    \sigma_c^{(k)}}_\vr,
    \label{eq:mathcaltappC}
\end{align}
where
$\varepsilon_{abc}$ denotes the Levi-Civita symbol 
for $a,b,c=x,y,z$.
Any $N$-qubit fully separable state can obey a certain tight bound
\begin{align}
    |\mathcal{T}(\vr)| \leq p_{\text{fs}}^{(N)},
\end{align}
where $p_{\text{fs}}^{(N)}$ can be computed analytically for $N=3$ and numerically for up to 
$N \leq 7$ and is, up to numerical precision, given by 
$p_{\text{fs}}^{(N)} = {N^2 \cot \left({\pi }/{N}\right)}/{3 \sqrt{3}}$.
Then violation implies the presence of multipartite entanglement.
\end{observation}
%%%%%%%%%%%%%%%%%%%%%%%%%%%%%%%%%%%%%%%%%%%%%%%%%%%%%
\begin{proof}
Here we give the derivation of Eq.~(\ref{eq:mathcaltappC}).
Let us begin by recalling
\begin{align}
    \mathcal{T}(\vr)
    = \int dU \,
    \tr\left[
    \vr U^{\otimes N}
    \mathcal{O}_{\mathcal{A}}
    (U^\dagger)^{\otimes N}
    \right],
    \qquad
    \mathcal{O}_{\mathcal{A}}
    = 
    \sum_{i < j < k}
    \mathcal{A}
    \left(
    \sigma_x^{(i)} \otimes \sigma_y^{(j)} \otimes \sigma_z^{(k)}
    \right),
\end{align}
where $\mathcal{A}$ represents a linear mapping that can make the antisymmetrization (or alternatization) by summing over even permutations and subtracting the sum over odd permutations.
More precisely, the observable can be rewritten as 
\begin{align}
    \mathcal{O}_{\mathcal{A}}
    = 
    \sum_{i < j < k}\
    \sum_{l,m,n=x,y,z}
    \varepsilon_{lmn}
    \sigma_l^{(i)} \otimes \sigma_m^{(j)} \otimes \sigma_n^{(k)}.
\end{align}
For instance, in the three-qubit system $ABC$, it is given by
\begin{align}\nonumber
    \mathcal{O}_{\mathcal{A}}
    &= \sigma_x^{(A)} \otimes \sigma_y^{(B)} \otimes \sigma_z^{(C)}
    + \sigma_y^{(A)} \otimes \sigma_z^{(B)} \otimes \sigma_x^{(C)}
    + \sigma_z^{(A)} \otimes \sigma_x^{(B)} \otimes \sigma_y^{(C)}\\
    \quad
    &-\sigma_x^{(A)} \otimes \sigma_z^{(B)} \otimes \sigma_y^{(C)}
    - \sigma_y^{(A)} \otimes \sigma_x^{(B)} \otimes \sigma_z^{(C)}
    - \sigma_z^{(A)} \otimes \sigma_y^{(B)} \otimes \sigma_x^{(C)}.
\end{align}
Then we have
\begin{align}\nonumber
        \mathcal{T}(\vr)
        &=
        \sum_{i < j < k}\
        \sum_{l,m,n=x,y,z}
        \varepsilon_{lmn}
        \left\{\int dU\,
        \tr\left[
        \vr_{ijk} U^{\otimes 3}
        \sigma_l^{(i)} \otimes \sigma_m^{(j)} \otimes \sigma_n^{(k)}
        (U^\dagger)^{\otimes 3}
        \right]
        \right\}\\
        &=\frac{1}{2^3}
        \sum_{i < j < k}\
        \sum_{l,m,n=x,y,z}
        \sum_{a,b,c=x,y,z}
        \varepsilon_{lmn}
        \xi_{abc}^{(i,j,k)}
        \int dU\,
        \tr[\sigma_a^{(i)} U \sigma_l^{(i)} U^\dagger]
        \tr[\sigma_b^{(j)} U \sigma_m^{(j)} U^\dagger]
        \tr[\sigma_c^{(k)} U \sigma_n^{(k)} U^\dagger],
\end{align}
where
$\vr_{ijk}$ is the three-qubit reduced state
of $\vr$ for $i,j,k=1,2,\ldots,N$
with the three-body correlation
$\xi_{abc}^{(i,j,k)} =
\tr[\vr_{ijk}\sigma_a^{(i)}\sigma_b^{(j)}\sigma_c^{(k)}]$
for $a,b,c=x,y,z$.

To evaluate the Haar unitary integral, we use the formula
\begin{align}
        \int dU\,
        \tr[\sigma_a U \sigma_x U^\dagger]
        \tr[\sigma_b U \sigma_y U^\dagger]
        \tr[\sigma_c U \sigma_z U^\dagger]
        = \frac{4}{3} \varepsilon_{abc},
\end{align}
which has been derived in Ref.~\cite{wyderka2023complete}.
Using this formula leads to
\begin{align}
    \mathcal{T}(\vr)
    =\frac{1}{2^3}\frac{4}{3}
        \sum_{i < j < k}\
        \sum_{l,m,n=x,y,z}
        \sum_{a,b,c=x,y,z}
        \varepsilon_{lmn}
        \xi_{abc}^{(i,j,k)}
        \varepsilon_{abc}
    =   \sum_{i < j < k}\
        \sum_{a,b,c=x,y,z}
        \xi_{abc}^{(i,j,k)}
        \varepsilon_{abc}.
\end{align}
Hence we can complete the proof.
\end{proof}

%%%%%%%%%%%%%%%%%%%%%%%%%%%%%%%%%%%%%%%%%%%%%%%%%%%%%

\noindent
\textbf{Remark C1.}
Here we explain how to derive the bound $p_{\text{fs}}^{(N)}$.
First, we note that the average $|\mathcal{T}(\vr)|$ is a convex function
of a quantum state.
Then it is enough to maximize the average for all $N$-qubit pure fully separable states:
$\ket{\Phi_{\text{fs}}}=\bigotimes_{i=1}^N \ket{\chi_i}$.
Each of single-qubit states $\ket{\chi_i}$ can be mapped
into points on the surface in the single-qubit Bloch sphere, which can be parameterized as
$\ex{\sigma_x}_{\chi_i}
= \cos (\theta_i)$,
$\ex{\sigma_y}_{\chi_i}
= \sin (\theta_i ) \cos(\phi_i)$,
and 
$\ex{\sigma_z}_{\chi_i}
=\sin (\theta_i) \sin (\phi_i)$
for $\chi_i = \ket{\chi_i}\! \bra{\chi_i}$.
Substituting these expressions into $|\mathcal{T}(\vr)|$
and performing its maximization over the parameters, we
can find the bound $p_{\text{fs}}^{(N)}$.
From numerical research up to $N\leq 7$, we collect
evidence that there may exist the \textit{tight} bound
\begin{align}
    p_{\text{fs}}^{(N)}
    =
    \frac{N^2 \cot \left(\frac{\pi }{N}\right)}{3 \sqrt{3}},
    \label{eq:criticalpointfullsep}
\end{align}
which may be obtained by
$\theta_i =2 \tan ^{-1}(\sqrt{2-\sqrt{3}})$
and 
$\phi_i= {2 \pi  i}/{N}$ for $i=1,2,\ldots, N$.
In Fig.~\ref{Fig3:Blochob3}, we illustrate the geometrical expressions of the points $\ket{\chi_i}$ on the surface in the single-qubit Bloch sphere
for $N=6, 20, 100$.
Note that we will give the analytical proof of the case with $N=3$ at
\textbf{Corollary} in the end of Appendix~\ref{ap:antisymmetricdet}.

%%%%%%%%%%%%%%%%%%%%%%%%%%%%%%%%%%%%%%%%%%%%%%%%%%%%%%%%%%%%%%%
\begin{figure}[t]
    \centering
    \includegraphics[width=1.0\columnwidth]{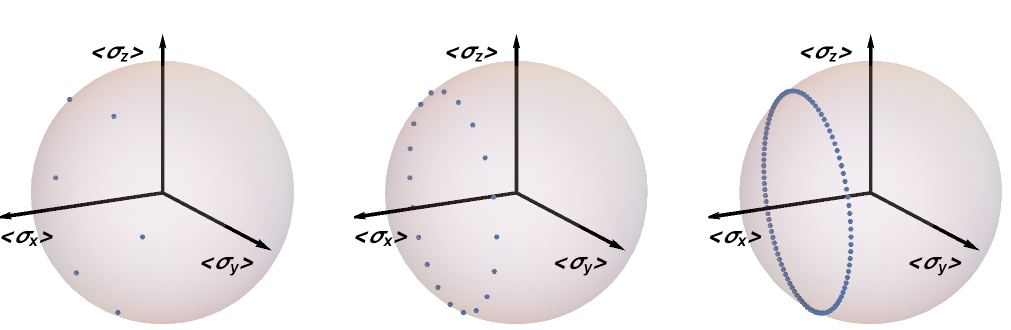}
    \caption{Geometry of $N$ single-qubit states $\ket{\chi_i}$ represented as (Blue) points on the surface in the single-qubit Bloch sphere, for $i=1,2,\ldots N$ and $N=6, 20, 100$.
    }
    \label{Fig3:Blochob3}
\end{figure}
%%%%%%%%%%%%%%%%%%%%%%%%%%%%%%%%%%%%%%%%%%%%%%%%%%%%%%%%%%%%%%%

\vspace{1em}
\noindent
\textbf{Remark C2.}
In Fig.~\ref{Fig4:casethreequbit}, we illustrate the criterion of Observation~\ref{ob:antisymmetricdet} for the state $\vr_{x,y}$ in Eq.~(\ref{eq:mixedteststate}) in the main text
for $N=4,5,6$ on the $x-y$ plane.

%%%%%%%%%%%%%%%%%%%%%%%%%%%%%%%%%%%%%%%%%%%%%%%%%%%%%%%%%%%%%%%
\begin{figure}[t]
    \centering
    \includegraphics[width=1.0\columnwidth]{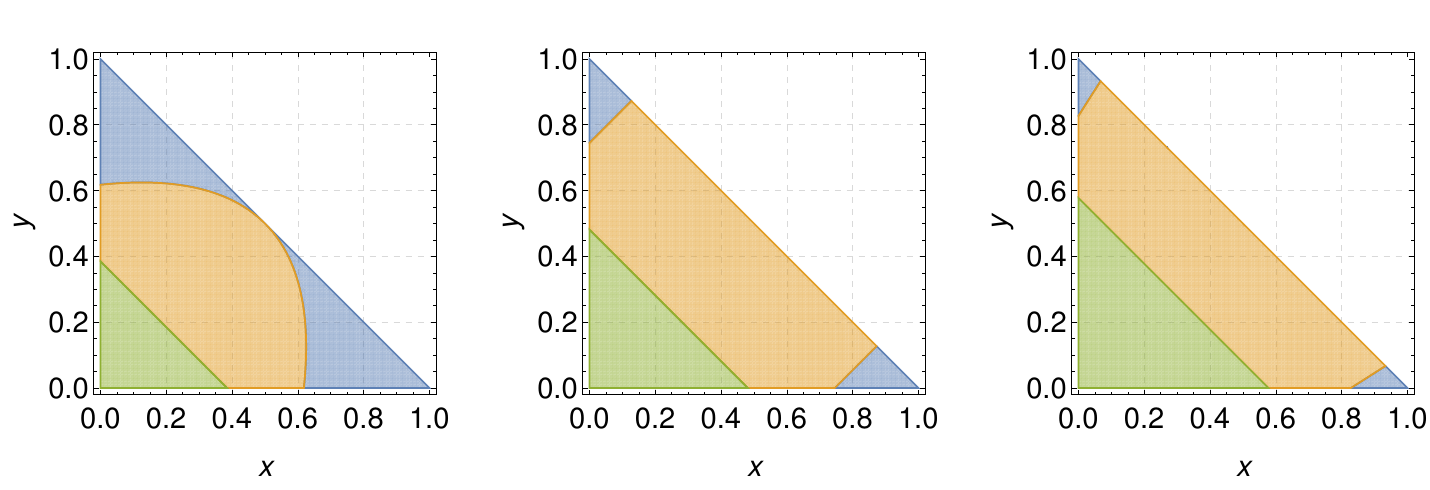}
    \caption{Entanglement criteria for the mixed state
    in Eq.~(\ref{eq:mixedteststate}) in the main text for $N=4,5,6$ in the $x-y$ plane.    
    The fully separable states are contained in Green area,
    which obeys all the optimal spin-squeezing
    inequalities (OSSIs) previously known with optimal measurement directions~\cite{toth2007optimal, toth2009spin} and also our criterion in Obs.~\ref{ob:antisymmetricdet} in the main text.
    Blue area corresponds to the spin-squeezed entangled states that can be detected by all OSSIs and Obs.~\ref{ob:antisymmetricdet}.
    Yellow area corresponds to the entangled states that cannot be detected by all OSSIs but can be detected by Obs.~\ref{ob:antisymmetricdet},
    thus marking the improvement of
    this paper compared with previous results.}
    \label{Fig4:casethreequbit}
\end{figure}
%%%%%%%%%%%%%%%%%%%%%%%%%%%%%%%%%%%%%%%%%%%%%%%%%%%%%%%%%%%%%%%

\vspace{1em}
\noindent
\textbf{Remark C3.}
Let us generalize Observation~\ref{ob:antisymmetricdet} in the main text by focusing only on three-particle systems.
We begin by denoting three-particle $d$-dimensional (three-qudit) operator as
\begin{align}
    W_S = \sum_{i,j,k}
    w_{ijk}
    s_{i} \otimes s_{j} \otimes s_{k},
\end{align}
for some given three-fold tensor $w_{ijk}$
and matrices $s_{i} \in \mathcal{H}_d$ with 
$s_{i} \neq \eins_d$.
If $d=2$,
$w_{ijk} = \varepsilon_{ijk}$,
and $s_i = \sigma_i$,
then it holds that $|\ex{W_S}| = |\mathcal{T}|$.
To proceed, we recall that a three-particle
state is called biseparable if
\begin{align}
    \vr_{\text{bs}}
= \sum_k p_k^A \vr_k^{A}\otimes \vr_k^{BC}
+ \sum_k p_k^B \vr_k^{B}\otimes \vr_k^{CA}
+ \sum_k p_k^C \vr_k^{C}\otimes \vr_k^{AB},
\end{align}
where and $p_k^X$ for $X=A,B,C$
are probability distributions.
The state is called genuinely multiparticle entangled if it cannot be written in the form of $\vr_{\text{bs}}$.
Now we will make the following:

%%%%%%%%%%%%%%%%%%%%%%%%%%%%%%%%%%%%%%%%%%%%%%%%%%%%%
\vspace{1em}
\noindent
\textbf{Lemma.}
\textit{For a three-qudit state $\vr_{ABC}$,
we denote 
the vector $\vec{s}_X = (s_i^X)$
and 
the matrix $S_{XY} = (s_{ij}^{XY})$
with
$s_{i}^{X} =\tr[\vr^{X} s_{i}]$ and
$s_{ij}^{XY} =\tr[\vr^{XY} s_{i} \otimes s_{j}]$,
where
$\vr_{X}, \vr_{XY}$ are marginal reduced states
of $\vr_{ABC}$ for $X, Y, Z = A, B, C$.
Any three-qudit fully separable state obeys
\begin{align}
    |\ex{W_S}| \leq
    \max_{X,Y,Z=A,B,C}
    \norm{\vec{s}_X}
    \norm{\vec{v}_{YZ}},
    \label{eq:fullysepAnti}
\end{align}
where
$\norm{\vec{v}}^2 = \sum_i v_i^2$
is the Euclidean vector norm of a vector $\vec{v}$
with elements $v_i$
and the vector 
$\vec{v}_{YZ}
=(v_{i}^{YZ})$
with 
$v_{i}^{YZ}
= \sum_{j,k}
s_j^Y s_k^Z w_{ijk}$.
Also, any three-qudit biseparable state obeys
\begin{align}
    |\ex{W_S}| \leq
    \max_{X,Y,Z=A,B,C}
    \sum_{i}
    \sigma_i(S_{XY})
    \sigma_i({Z}^\ast),
    \label{eq:bisepAnti}
\end{align}
where
$\sigma_i(O)$
are singular values of a matrix $O$
in decreasing order
and the matrix
${Z}^\ast
=(z_{ij}^\ast)$
with 
${z}_{ij}^\ast
= \sum_{k} s_k^Z w_{ijk}$.}

%%%%%%%%%%%%%%%%%%%%%%%%%%%%%%%%%%%%%%%%%%%%%%%%%%%%%%%%%%%%%%%%
\begin{proof}
Since $|\ex{W_S}|$ is a convex function for a state,
it is sufficient to prove the cases of pure states.
First, let us consider a pure fully separable state
$\vr^{A} \otimes \vr^B \otimes \vr^{C}$. Then we have
\begin{align}
    \ex{W_S}
    =\sum_{i,j,k}
    w_{ijk}
    \tr[\vr^{A} \otimes \vr^B \otimes \vr^{C}
    s_{i} \otimes s_{j} \otimes s_{k}]
    =\sum_i s_{i}^{A} \sum_{j,k} s_{j}^{B} s_{k}^{C} w_{ijk}
    =\sum_i s_{i}^{A} v_{i}^{BC}
    \leq \norm{\vec{s}_A} \norm{\vec{v}_{BC}},
\end{align}
where we used the Cauchy–Schwarz inequality to derive the inequality.
Similarly, we can have cases that correspond to $\vec{s}_B$ and $\vec{s}_C$.

Second, let us consider a pure biseparable state for a fixed bipartition
$XY|Z$.
For a case $AB|C$, we have
\begin{align}
\ex{W_S}
=\sum_{i,j}
s_{ij}^{AB} \sum_k s_k^C w_{ijk} 
=\sum_{i,j}
s_{ij}^{AB} {c}_{ij}^\ast
=\tr [
S_{AB} ({{C}^\ast})^\top ]
\leq
\sum_{i}
\sigma_i(S_{AB})
\sigma_i({C}^\ast),
\end{align}
where we used 
von Neumann’s trace inequality~\cite{mirsky1975trace}.
Similarly, we can obtain the bounds for the other
bipartitions $B|CA$ and $C|AB$.
Hence we can complete the proof.
\end{proof}

%%%%%%%%%%%%%%%%%%%%%%%%%%%%%%%%%%%%%%%%%%%%%%%%%%%%%
\noindent
\textbf{Corollary.}
\textit{Consider the case where
$d=2$,
$w_{ijk} = \varepsilon_{ijk}$,
and $s_i = \sigma_i$.
Any three-qubit fully separable state obeys
$
    |\ex{W_S}| \leq 1.
$
Also, any three-qubit biseparable state obeys
$|\ex{W_S}| \leq 2.$
}
\begin{proof}
    To show these, without loss of generality,
    we can take
    $\vr^{C} = \ket{0}\!\bra{0}$.
    This can lead to that
    $\vec{v}_{BC} = (s_2^B, -s_1^B, 0)$.
    For single-qubit pure states, we have that
    $\norm{\vec{s}_A} = 1$
    and
    $\norm{\vec{v}_{BC}} \leq 1.$
    Thus we can show the the fully separable bound.
    For the biseparable bound, since
    $\sigma_1({C}^\ast) = \sigma_2({C}^\ast) =1$
    and $\sigma_3({C}^\ast) =0$,
    we can immediately
    find that    
    $\sigma_1(S_{AB}) + \sigma_2(S_{AB}) \leq 2$
    for all pure $\vr^{AB}$.
\end{proof}

%%%%%%%%%%%%%%%%%%%%%%%%%%%%%%%%%%%%%%%%%%%%%%%%%%%%%
\section{Detailed discussion of Observation~\ref{ob:twoensembles}}\label{ap:twoensembles}
\begin{observation}
For a $2N$-qubit state $\vr_{AB}$
with the permutationally symmetric reduced states,
any separable $\vr_{AB}$ obeys
\begin{align}
    \mathcal{G}_{AB}^{(2)}
    + \mathcal{J}_A^{(1)}
    + \mathcal{J}_B^{(1)}
    - \mathcal{J}_A^{(1)}
    \mathcal{J}_B^{(1)}
    \leq 1,
\end{align}
where
$g = ({3}/{N^2})^2$
and 
$(\alpha, \beta, \gamma) = (0,12/N^2,0)$.
\end{observation}
%%%%%%%%%%%%%%%%%%%%%%%%%%%%%%%%%%%%%%%%%%%%%%%%%%%%%
\begin{proof}
We begin by writing
\begin{subequations}
    \begin{align}
    \mathcal{G}_{AB}^{(r)}
    &= g\int dU_A \int dU_B \,
    [\eta_{U_{AB}}]^r,
    &
    \eta_{U_{AB}}
    &= \va{J_{z}^{+}}_{U_{AB}} - \va{J_{z}^{-}}_{U_{AB}},\\
    \mathcal{J}_X^{(r)} (\vr_X)
    &= \int dU_X \, [f_U (\vr_X)]^r,
    &
    f_U (\vr_X)
    &= \alpha \va{J_{z, X}}_{U_X} + \beta \ex{J_{z, X}}_{U_X}^2 + \gamma,
\end{align}
\end{subequations}
where
\begin{subequations}
\begin{align}
    \ex{J_{z}^{\pm}}_{U_{AB}}
    &= \tr\left[
    \vr_{AB} U_{AB}^{\otimes N} J_{z}^{\pm}
    (U_{AB}^\dagger)^{\otimes N}
    \right],
    \qquad
    \va{J_{z}^{\pm}}_{U_{AB}}
    = \ex{(J_{z}^{\pm})^2}_{U_{AB}}
    -\ex{J_{z}^{\pm}}_{U_{AB}}^2,\\
    J_{z}^{\pm}
    &= J_{z, A} \pm J_{z, B},
    \qquad
    J_{z, X} =
    \frac{1}{2}\sum_{i=1}^N \sigma_{z}^{(X_i)},
    \qquad
    U_{AB} = U_A \otimes U_B.
\end{align}    
\end{subequations}
Then we can have
\begin{subequations}
\begin{align}
    \ex{J_{z}^{\pm}}_{U_{AB}}^2
    &= \ex{J_{z, A}}_{U_{A}}^2 + \ex{J_{z, B}}_{U_{B}}^2
    \pm 2\ex{J_{z, A}}_{U_{A}} \ex{J_{z, B}}_{U_{B}},\\
    \va{J_{z}^{\pm}}_{U_{AB}}
    &= \va{J_{z, A}}_{U_{A}} + \va{J_{z, B}}_{U_{B}} \pm 2 \left[
    \ex{J_{z, A} \otimes J_{z, B}}_{U_{AB}}
    - \ex{J_{z, A}}_{U_{A}} \ex{J_{z, B}}_{U_{B}}
    \right],\\
    \eta_{U_{AB}}
    &= 4 \left[
    \ex{J_{z, A} \otimes J_{z, B}}_{U_{AB}}
    - \ex{J_{z, A}}_{U_{A}} \ex{J_{z, B}}_{U_{B}}
    \right].
\end{align}    
\end{subequations}
Let us evaluate the form of $\mathcal{G}_{AB}^{(2)}(\vr_{AB})$.
Applying the assumption that $\vr_A$ and $\vr_B$ are permutationally symmetric,
we can further simplify the form of $\eta_{U_{AB}}$
\begin{align} \nonumber
    \eta_{U_{AB}}
    &=4 \frac{1}{2} \frac{1}{2}
    \left\{
    \sum_{i,j=1}^N
    \tr[\vr_{A_i B_j}
    U_A \sigma_{z}^{(A_i)} U_A^\dagger \otimes
    U_B \sigma_{z}^{(B_j)} U_B^\dagger    ]
    -
    \sum_{i,j=1}^N
    \tr[\vr_{A_i}
    U_A \sigma_{z}^{(A_i)} U_A^\dagger]
    \tr[\vr_{B_i}
    U_B \sigma_{z}^{(B_i)} U_B^\dagger]
    \right\}\\
    &=\frac{N^2}{4} \sum_{p,q=x,y,z}
    C_{pq} \mathcal{O}_{U_{A}}^{(p)}\mathcal{O}_{U_{B}}^{(q)},
\end{align}
where
the covariance matrix $C=(C_{pq})$ is given by
\begin{align}
    C_{pq}
    =\tr[\vr_{A_i B_j} \sigma_p^{(A_i)} \otimes \sigma_q^{(B_j)}]
    -\tr[\vr_{A_i} \sigma_p^{(A_i)}]
    \tr[\vr_{B_j} \sigma_q^{(B_j)}]
    = t_{pq} - a_p b_q,
\end{align}
for the two-qubit reduced state $\vr_{A_i B_j} = \tr_{\overline{ij}}(\vr_{AB})$ such that both particles are still spatially separated, defined in $\mathcal{H}_{A_i} \otimes \mathcal{H}_{B_j}$.
Here we denote that
\begin{align}
    \mathcal{O}_{U_{X}}^{(p)} =
    \tr\left[\sigma_p^{(X_i)}
    U_{X}
    \sigma_{z}^{(X_i)}
    U_{X}^\dagger
    \right],
\end{align}
for $X_i = A_i, B_i$.
Notice that $C_{pq}$ and $\mathcal{O}_{U_{X}}^{(p)}$ are independent of indices $i,j$ due to the permutational symmetry.

To avoid confusion, we have to stress that the above covariance matrix
$C=(C_{pq})$ is different from Eq.~(\ref{eq:covappA}) in Appendix~\ref{ap:neccesuffi}
in general.
If the spatially-separated reduced state $\vr_{A_i B_j}$ is also permutationally symmetric,
both are the same, but here we do not require the assumption.

Using the formula in Eq.~(\ref{eq:haarijappa}) in Appendix~\ref{ap:neccesuffi}, we have that
\begin{align}
    \mathcal{G}_{AB}^{(2)}(\vr_{AB})
    =g \frac{N^4}{4^2}
    \sum_{p,q,r,s=x,y,z}
    C_{pq} C_{rs}
    \int dU_A \,
    \mathcal{O}_{U_{A}}^{(p)}
    \mathcal{O}_{U_{A}}^{(r)}
    \int dU_B \,
    \mathcal{O}_{U_{B}}^{(q)}
    \mathcal{O}_{U_{B}}^{(s)}
    =\sum_{p,q=x,y,z}
    C_{pq}^2,
\end{align}
where we set that $g = ({3}/{N^2})^2$.
Also, since
$\ex{J_{z, A}}_{U_{A}} = (N/4)\sum_{p=x,y,z}a_p \mathcal{O}_{U_{A}}^{(p)}$ and $\beta= 12/N^2$,
we can find 
\begin{align}
    \mathcal{J}_A^{(1)}(\vr_A)
    =\beta \frac{N^2}{4^2}\sum_{p,q=x,y,z}
    a_p a_q
    \int dU_A \,
    \mathcal{O}_{U_{A}}^{(p)}
    \mathcal{O}_{U_{A}}^{(q)}
    =\sum_{p=x,y,z}a_p^2,
\end{align}
as well as $\mathcal{J}_B^{(1)}(\vr_B) = \sum_{p=x,y,z}b_p^2$.
In summary, for the choice $g = ({3}/{N^2})^2$
and 
$(\alpha, \beta, \gamma) = (0,12/N^2,0)$, we have that
\begin{align}
    \mathcal{G}_{AB}^{(2)}
    + \mathcal{J}_A^{(1)}
    + \mathcal{J}_B^{(1)}
    - \mathcal{J}_A^{(1)}
    \mathcal{J}_B^{(1)}
    =\sum_{p,q=x,y,z}C_{pq}^2
    + \sum_{p=x,y,z} (a_p^2+b_p^2)
    -\sum_{p,q=x,y,z} a_p^2 b_q^2.
    \label{eq:evaluatedform1}
\end{align}

To derive the entanglement criterion, we rewrite the right-hand-side in Eq.~(\ref{eq:evaluatedform1}) as
\begin{align}
    \mathcal{G}_{AB}^{(2)}
    + \mathcal{J}_A^{(1)}
    + \mathcal{J}_B^{(1)}
    - \mathcal{J}_A^{(1)}
    \mathcal{J}_B^{(1)}
    =\sum_{p,q=x,y,z} (t_{pq}^2-2 a_p b_q t_{pq})
    + \sum_{p=x,y,z} (a_p^2+b_p^2),
    \label{eq:evaluatedlhsapp}
\end{align}
where we use that
$C_{pq}^2
= t_{pq}^2
+ a_p^2 b_q^2
-2 a_p b_q t_{pq}$.
To proceed further, we recall the separability criterion
presented in Ref.~\cite{liu2022characterizingcorr} (see, Proposition 5):
if a bipartite quantum state $\vr_{XY}$ is separable, then it obeys that
\begin{align}
    \tr(\vr_{XY}^2)+ \tr(\vr_X^2) + \tr(\vr_Y^2) -2 \tr[\vr_{XY}(\vr_X \otimes \vr_Y)]
    \leq 1.
\end{align}
If $\vr_{XY}$ is a two-qubit state, we can rewrite this inequality as
\begin{align}
\sum_{i,j=x,y,z} (z_{ij}^2-2 x_i y_j z_{ij}) +
    \sum_{i=x,y,z} ( x_i^2+ y_i^2)
    \leq 1,
\end{align}
where
$x_i = \tr{(\vr_X \sigma_i)}$,
$y_i = \tr{(\vr_Y \sigma_i)}$,
and
$z_{ij} = \tr{(\vr_{XY} \sigma_i \otimes \sigma_j)}$.
Let us apply this criterion to Eq.~(\ref{eq:evaluatedlhsapp}).
Exchanging the symbols 
\begin{align}
    x_i \longleftrightarrow a_p,
    \ \ \ \ 
    y_i \longleftrightarrow b_p
    \ \ \ \ 
    z_{ij} \longleftrightarrow t_{pq},
\end{align}
we can connect this criterion to Eq.~(\ref{eq:evaluatedlhsapp})
and arrive at the inequality in Observation~\ref{ob:twoensembles}.
Hence we can complete the proof.
\end{proof}

\noindent
\textbf{Remark D1.}
The right-hand-side in Eq.~(\ref{eq:evaluatedform1}),
that is,
the right-hand-side in Eq.~(\ref{eq:twoensembledetection})
in Observation~\ref{ob:twoensembles} in the main text,
can be rewritten as
\begin{align}
    \mathcal{G}_{AB}^{(2)}
    + \mathcal{J}_A^{(1)}
    + \mathcal{J}_B^{(1)}
    - \mathcal{J}_A^{(1)}
    \mathcal{J}_B^{(1)}
    &=\frac{1}{N^4}
    \left\{
    \sum_{p, q} 
    \eta_{pq}^2
    + 4 N^2
    \sum_{p}
    \left[
    \ex{J_{p, A}}^2 +  \ex{J_{p, B}}^2
    \right]
    - 16 \sum_{p, q}
    \ex{J_{p, A}}^2 \ex{J_{q, B}}^2
    \right\},
\end{align}
where
\begin{subequations}
    \begin{align}
    \sum_{p=x,y,z} a_p^2
    &= \left(\frac{2}{N}\right)^2
    \sum_{p=x,y,z} \ex{J_{p, A}}^2,
    \\
    \sum_{p=x,y,z} b_p^2
    &= \left(\frac{2}{N}\right)^2
    \sum_{p=x,y,z} \ex{J_{p, B}}^2,
    \\
    \sum_{p, q =x,y,z} C_{pq}^2
    &=\left(\frac{1}{N^2}\right)^2
    \sum_{p, q=x,y,z} 
    \left[\va{J_{p}^{+}}
    - \va{J_{q}^{-}}\right]^2
    \equiv
    \left(\frac{1}{N^2}\right)^2
    \sum_{p, q=x,y,z} 
    \eta_{pq}^2,
\end{align}
\end{subequations}
and
$\eta_{pq} \equiv
\va{J_{p}^{+}}    - \va{J_{q}^{-}}$.

\vspace{1em}
\noindent
\textbf{Remark D2.}
Here we consider the generalization of Observation~\ref{ob:twoensembles} in the main text to $m$ ensembles for $m \geq 3$.
For that, let us define a quantum state $\vr \in \mathcal{H}_1 \otimes \cdots \otimes \mathcal{H}_m$, where $\mathcal{H}_X = \mathcal{H}_2^{\otimes N}$ for $X=1,\ldots,m$.
Now it is essential to notice that the left-hand-side in Observation~\ref{ob:twoensembles} can be available for any two-pair in $m$ ensembles.
Then, let us define the average over all pairs
\begin{align}
    \mathcal{P}(\vr)
    =\frac{2}{m(m-1)}
    \sum_{X<Y}
    \mathcal{G}_{XY}^{(2)}
    + \mathcal{J}_X^{(1)}
    + \mathcal{J}_Y^{(1)}
    - \mathcal{J}_X^{(1)}
    \mathcal{J}_Y^{(1)},
\end{align}
for $X,Y=1,2,\ldots, m$.
Now we can formulate the following.

\vspace{1em}
\noindent
\textbf{Remark D3.}
\textit{For this $mN$-qubit state $\vr$ consisting of the $m$ ensembles of $N$ spin-$\frac{1}{2}$ particles, if each $N$-qubit ensemble is permutationally symmetric,
then any fully separable $\vr$ obeys
\begin{align}
    \mathcal{P}(\vr)
    \leq 1,
\end{align}
where
$g = ({3}/{N^2})^2$
and 
$(\alpha, \beta, \gamma) = (0,12/N^2,0)$.
}

\begin{proof}
    In general, if a multipartite state $\vr$ is fully separable,
    then all the bipartite reduced states
    are clearly separable.
    For such separable reduced states,
    Observation~\ref{ob:twoensembles} in the main text holds.
    Thus we can complete the proof.    
\end{proof}

\vspace{1em}
\noindent
\textbf{Remark D4.}
Let us test our criterion in Observation~\ref{ob:twoensembles} in the main text with the Dicke state as a bipartite state.
The $N_{AB}$-qubit Dicke state with $m_{AB}$ excitations is defined as
\begin{equation}
    \ket{{N_{AB},m_{AB}}}=\binom{N_{AB}}{m_{AB}}^{-\frac{1}{2}}\sum_{m_{AB}}
    \mathcal{P}_{m_{AB}}
    (\ket{1_{1},\ldots,1_{m_{AB}},
    0_{m_{AB}+1},\ldots,0_{N_{AB}}}),
\end{equation}
where
$\{\mathcal{P}_{m_{AB}}\}$ is the set of
all distinct permutations in the qubits.
Applying the Schmidt decomposition to the Dicke state, one can have
\begin{align}
    \ket{{N_{AB},m_{AB}}}=\sum_{m=0}^{N_{AB}}
    \lambda_{m}
    \ket{{N_A, m_A}}
    \otimes
    \ket{{N_B, m_B}},
\end{align}
where $N_A+N_B = N_{AB}$,
$m_A + m_B = m_{AB}$,
and $m = m_A$.
Here, the Schmidt coefficients $\lambda_m$
are given by
\begin{align}
    \lambda_m =
    \binom{N_{AB}}{m_{AB}}^{-\frac{1}{2}}
    \binom{N_A}{m_A}^{\frac{1}{2}}
    \binom{N_B}{m_B}^{\frac{1}{2}}.
\end{align}
The states
$\ket{{N_A,m_A}}$ and $\ket{{N_B, m_B}}$
are permutationally symmetric states,
for details, see~\cite{stockton2003characterizing, toth2007detection}.

Let us consider the case where
$N_A = N_B = N_{AB}/2$,
and
$m_{AB} = N_{AB}/2$.
Then we have
\begin{subequations}
\begin{align}
    &\ex{J_{p, A}}
    = \ex{J_{p, B}}
    = 0,
    \ \ \text{for}
    \ \ p=x,y,z,
    \\
    &\ex{J_z^2} = \va{J_{z}^{+}} =0,
    \\
    & \va{J_{z}^{-}} = -4 \ex{J_{z,A} \otimes J_{z,B}} = \frac{N_{AB}^2}{4(N_{AB}-1)},
    \\
    &\va{J_{x}^{+}}
    =\va{J_{y}^{+}}
    =\frac{N_{AB}}{4}
    \left(\frac{N_{AB}}{2}+1
    \right),
    \\
    &\va{J_{x}^{-}}
    =\va{J_{y}^{-}}
    =\frac{N_{AB}}{8}
    \frac{N_{AB}-2}{N_{AB}-1},
\end{align}
\end{subequations}
where we used the results in Ref.~\cite{vitagliano2023number}.
Then we have the values of
$
    \va{J_{p}^{\pm}}.
$
In this paper, we set $N_{AB} = 2N$.
This results in
\begin{align}
    \mathcal{G}_{AB}^{(2)}
    + \mathcal{J}_A^{(1)}
    + \mathcal{J}_B^{(1)}
    - \mathcal{J}_A^{(1)}
    \mathcal{J}_B^{(1)}
    =
    \frac{1}{N^4}
    \sum_{p, q=x,y,z} 
    \left[\va{J_{p}^{+}}
    - \va{J_{q}^{-}}\right]^2
    =
    \frac{6 N^4-2 N^3+1}{(1-2 N)^2 N^2}.
\end{align}
The right-hand side monotonically decreases as $N$ increases, and it becomes $3/2$ when $N \to \infty$.
Therefore the pure $2N$-qubit Dicke states can be detected in any $N$.

Finally, let us consider the case where the global depolarizing channel with noise $p$ influences the state as follows:
$\vr_D \to \vr_D^\prime = p \vr_D + (1-p)\vr_{\text{mm}}$
for
$\vr_D = \ket{{N_{AB},m_{AB}}} \! \bra{{N_{AB},m_{AB}}}$
and
the maximally mixed state $\vr_{\text{mm}}$.
This noise effects can change 
$\va{J_{l}^{\pm}}$ as follows:
\begin{subequations}
\begin{align}
    \va{J_{x/y}^{+}}_{\vr_D}
    &\to
    \va{J_{x/y}^{+}}_{\vr_D^\prime}
    =
    \frac{N_{AB}}{4}
    + p \left[
    \va{J_{x/y}^{+}}_{\vr_D}
    -\frac{N_{AB}}{4}
    \right]
    =\frac{N(1+Np)}{2},
    \\
    \va{J_{x/y}^{-}}_{\vr_D}
    &\to
    \va{J_{x/y}^{-}}_{\vr_D^\prime}
    =
    \frac{N_{AB}}{4}
    + p \left[
    \va{J_{x/y}^{-}}_{\vr_D}
    -\frac{N_{AB}}{4}
    \right]
    =\frac{N [N (2-p)-1]}{2 (2N-1)},
    \\
    \va{J_{z}^{+}}_{\vr_D}
    &\to
    \va{J_{z}^{+}}_{\vr_D^\prime}
    = \frac{N_{AB}(1-p)}{4}
    = \frac{N(1-p)}{2},
    \\
    \va{J_{z}^{-}}_{\vr_D}
    &\to
    \va{J_{z}^{-}}_{\vr_D^\prime}
    = \frac{N_{AB}}{4}
    + p \left[
    \va{J_{z}^{-}}_{\vr_D}
    -\frac{N_{AB}}{4}
    \right]
    = \frac{N(2N+p-1)}{2(2N-1)}.
\end{align}
\end{subequations}
This leads to
\begin{equation}
    \mathcal{G}_{AB}^{(2)}
    + \mathcal{J}_A^{(1)}
    + \mathcal{J}_B^{(1)}
    - \mathcal{J}_A^{(1)}
    \mathcal{J}_B^{(1)}
    = \frac{\left(6 N^4-2 N^3+1\right) p^2}{(1-2 N)^2 N^2}.
\end{equation}
Then we can find that the separability bound in Observation~\ref{ob:twoensembles} is violated when $p>p^*(N)$ for the critical point
\begin{equation}
    p^*(N)
    =
    \frac{N (2 N-1)}{\sqrt{6 N^4-2 N^3+1}}.
\end{equation}
In Fig.~\ref{Fig5:criticalpoint}, we illustrate the behavior of the critical point depending on $N$.
In the limit $N\to \infty$, this point becomes $p^* \to \sqrt{2/3}$.

%%%%%%%%%%%%%%%%%%%%%%%%%%%%%%%%%%%%%%%%%%%%%%%%%%%%%%%%%%%%%%%
\begin{figure}[t]
    \centering
    \includegraphics[width=0.5\columnwidth]{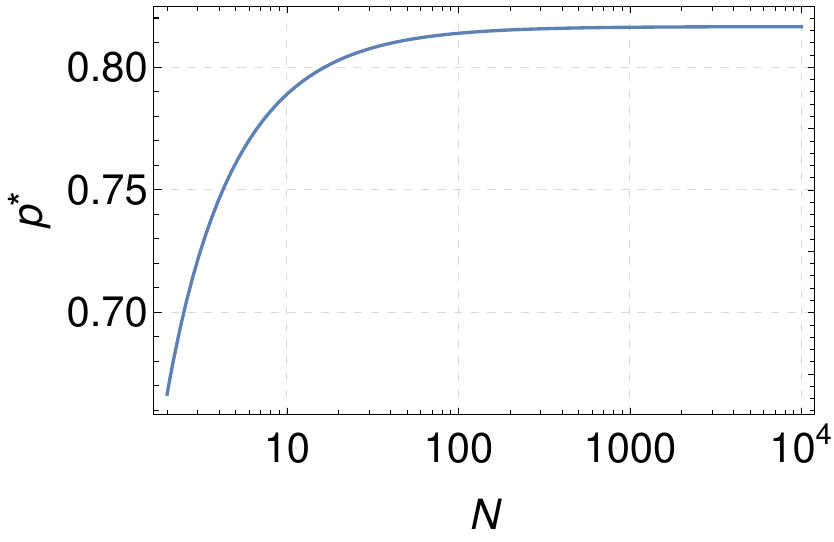}
    \caption{Linear-Log plot of the critical point $p^*(N)$ discussed in Remark D4.
    }
    \label{Fig5:criticalpoint}
\end{figure}

%%%%%%%%%%%%%%%%%%%%%%%%%%%%%%%%%%%%%%%%%%%%%%%%%%%%%
\section{Entanglement detection with finite statistics}\label{ap:finitestatistics}
Here we will discuss the estimation of the statistical error of the moments from collective randomized measurements. We will determine the number of measurements required for reliable entanglement detection, following similar discussions given in the previous works~\cite{ketterer2022statistically,wyderka2023probing,liu2023characterizing,wyderka2023complete,cieslinski2023analysing}.

In randomized measurement schemes, the total number of measurements is denoted as $M_{\text{tot}} = M \times K$.
Here $M$ is the number of random unitaries, and $K$ is the number of measurements for a fixed unitary.
For the moment $\mathcal{J}^{(r)} = \int dU \, [f_U]^r$ defined in Eq.~(\ref{eq:momentJr}) in the main text, the unbiased estimator can be given by $\Tilde{\mathcal{J}}^{(r)} = (1/M) \sum_{i=1}^M [\Tilde{f}_r]_i$,
that is, $\mathbb{E}_U \mathbb{E}[\Tilde{\mathcal{J}}^{(r)}] = \mathcal{J}^{(r)}$.
Here, the subscript $i$ in $[\Tilde{f}_r]_i$ refers to the measurement setting,
$\mathbb{E}_U$ represents the average over collective local random unitaries,
and $\Tilde{f}_r$ is the unbiased estimator of $[f_U]^r$, that is, $\mathbb{E}[\Tilde{f}_r] = [f_U]^r$.

To quantify how much the estimator $\Tilde{\mathcal{J}}^{(r)}$ deviates from the $\mathcal{J}^{(r)}$,
let us consider the inequality
\begin{equation}
    \text{Prob}(
    \Tilde{\mathcal{J}}^{(r)} - \mathcal{J}^{(r)}
    \geq \delta_{\text{error}}) \leq \alpha_{\text{ssl}},
\end{equation}
where
$\delta_{\text{error}}$ is called the error or accuracy,
$\alpha_{\text{ssl}}$ the statistical significance level,
and $\gamma_{\text{cl}} = 1-\alpha_{\text{ssl}}$ the confidence level.
In the case where $[\Tilde{f}_r]_i$ cannot be assumed to be i.i.d.~random variables,
we can employ the so-called (one-sided) Chebyshev-Cantelli inequality
\begin{equation}
    \text{Prob}\left(
    \Tilde{\mathcal{J}}^{(r)} - \mathcal{J}^{(r)}
    \geq \delta_{\text{error}}
    \right) \leq \frac{\text{Var}(\Tilde{\mathcal{J}}^{(r)})}{\text{Var}(\Tilde{\mathcal{J}}^{(r)}) + \delta_{\text{error}}^2},
    \label{eq:Cantelli}
\end{equation}
where
$\text{Var}(\Tilde{\mathcal{J}}^{(r)})$ denotes the variance of the estimator $\Tilde{\mathcal{J}}^{(r)}$.
Based on this inequality, one can determine the total number of measurements $M_{\text{tot}} = M \times K$ required for entanglement detection, for a fixed error and confidence level.
Since it holds that $\text{Var}(\Tilde{\mathcal{J}}^{(r)}) = (1/M^2)\sum_{i=1}^M \text{Var}([\Tilde{f}_r]_i)$,
the main task is to evaluate the variance $\text{Var}([\Tilde{f}_r]_i)$.

In the following, as a simple example, we will particularly focus on Observation~\ref{ob:singletcriterion} in the main text: the moment $\mathcal{J}^{(1)} = \int dU \, [3\va{J_z}_U]$ is equal to $\sum_{l=x,y,z}\va{J_l}$, and any $N$-qubit fully separable state obeys $\mathcal{J}^{(1)}(\vr) \geq N/2$. Now the variance is written as
$\text{Var}(\Tilde{f}_1) = \mathbb{E}_U \mathbb{E} [(\Tilde{f}_1)^2] - [\mathcal{J}^{(1)}]^2$, and the explicit form of $\mathbb{E} [(\Tilde{f}_1)^2]$ can be given by
\begin{equation}
    \mathbb{E} [(\Tilde{f}_1)^2]
    = 9 \left\{ c_1(K) \ex{J_z^4}_{U}
    + c_2(K) \ex{J_z^3}_{U} \ex{J_z}_{U}
    + c_3(K) \ex{J_z^2}_{U}^2
    + c_4(K) \ex{J_z^2}_{U} \ex{J_z}_{U}^2
    + c_5(K) \ex{J_z}_{U}^4
    \right\},
\end{equation}
where
$c_1(K) = 1/K$,
$c_2(K) = -4/K$,
$c_3(K) = [(K-1)^2+2]/[K(K-1)]$,
$c_4(K) = -2(K-2)(K-3)/[K(K-1)]$,
and
$c_5(K) = (K-2)(K-3)/[K(K-1)]$.
This expression can be derived using the result in Ref.~\cite{bonsel2023error} that coincides with \cite{o2014some}.

Let us consider the statistically significant test with the family of the states
\begin{equation}
    \vr_p = (1-p) \vr_{\text{singlet}} + p \frac{\eins}{2^N},
\end{equation}
where $\vr_{\text{singlet}}$ denotes the $N$-qubit many-body spin singlet state, discussed in the main text.
Since this state obeys $\ex{J_l^k}_{\vr_{\text{singlet}}} = 0$ for any $k$,
we have that
$\ex{J_l}_{\vr_p} = 0$,
$\ex{J_l^2}_{\vr_p} = Np/4$,
and therefore,
$\mathcal{J}^{(1)}(\vr_p) = 3Np/4$.
Thus, the state $\vr_p$ violates the separability bound in
Observation~\ref{ob:singletcriterion}
when $p$ becomes smaller than the critical point $p_\text{sep} = 2/3$,
which is independent of $N$.
From a straightforward calculation, we can obtain
\begin{equation}
    \text{Var}(\Tilde{\mathcal{J}}^{(1)})
    = \frac{9 N p \{3 N (p-1) + 2 - K [N (p-3)+2]\}}{[16 (K-1) K] M},
\end{equation}
where we used $\ex{J_l^4}_{\vr_p} = N(3N-2)p/16$.
Rearranging the Chebyshev-Cantelli inequality in Eq.~(\ref{eq:Cantelli})
and requiring that the confidence $1- \text{Prob}(\Tilde{\mathcal{J}}^{(1)} - \mathcal{J}^{(1)} \geq \delta_{\text{error}})$ is at least $\gamma_{cl}$,
we have
\begin{equation}
    \delta_{\text{error}} = \sqrt{\frac{\gamma_{\text{cl}}}{1-\gamma_{\text{cl}}} \text{Var}(\Tilde{\mathcal{J}}^{(1)})}.
\end{equation}
Since the variance $\text{Var}(\Tilde{\mathcal{J}}^{(1)})$ can monotonically increase for large $p$, the worst-case error in the estimation is given by $p=1$.
In this case, we have
\begin{equation}
    M = \frac{9 \gamma_{\text{cl}}  N [K (N-1)+1]}{8 (1-\gamma_{\text{cl}}) (K-1) K
    \delta_{\text{error}}^2}.
\end{equation}
To proceed, let us set the error as
$\delta_{\text{error}} = \min_{\vr_{\text{sep}} \in \text{SEP}} \mathcal{J}^{(1)}(\vr_{\text{sep}}) - \mathcal{J}^{(1)}(\vr_p) = N/2 - 3Np/4$.
By minimizing the value of $M_{\text{tot}}= M(K) \times K$ with respect to $K$
for a fixed $N$, we can thus find the optimal number of measurements.
In Fig.~\ref{Fig6:numbermeasurement}, we illustrate the necessary number of measurements for $N = 100$ and the confidence level $\gamma_{\text{cl}} = 0.95$.

%%%%%%%%%%%%%%%%%%%%%%%%%%%%%%%%%%%%%%%%%%%%%%%%%%%%%%%%%%%%%%%
\begin{figure}[t]
    \centering
    \includegraphics[width=0.5\columnwidth]{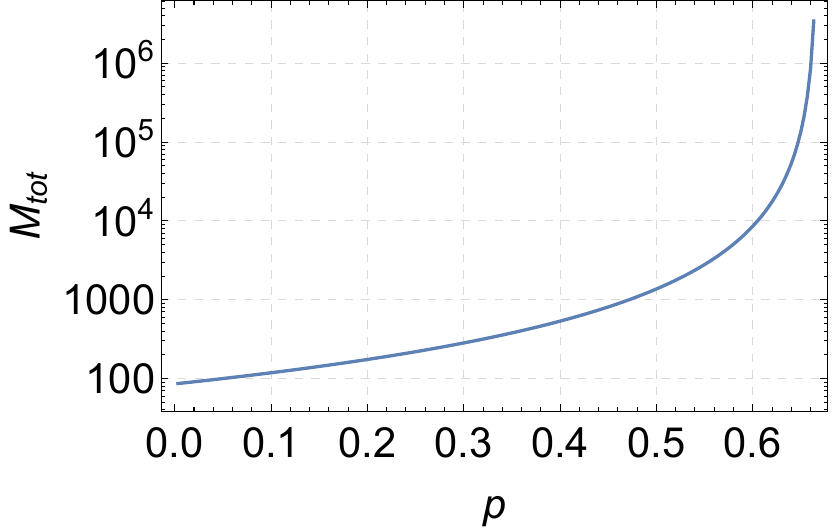}
    \caption{Log plot of the total number of measurements $M_{\text{tot}}$
    obtained from the Chebyshev-Cantelli inequality required to certify the
    violation of the separability bound in Observation~\ref{ob:singletcriterion}
    in the main text of $\vr_p$, for
    $N=10^2$ and confidence level $\gamma_{\text{cl}} = 0.95$.}
    \label{Fig6:numbermeasurement}
\end{figure}
%%%%%%%%%%%%%%%%%%%%%%%%%%%%%%%%%%%%%%%%%%%%%%%%%%%%%%%%%%%%%%%

%%%%%%%%%%%%%%%%%%%%%%%%%%%%%%%%%%%%%%%%%%%%%%%%%%%%%

%\bibliography{ref.bib}

\end{document}